\documentclass[
	journal
]{IEEEtran}

\renewcommand{\markboth}[1]{\renewcommand{\leftmark}{#1}\renewcommand{\rightmark}{#1}}
\IEEEoverridecommandlockouts 
\makeatletter
\def\markboth#1#2{\def\leftmark{\@IEEEcompsoconly{\sffamily}\MakeUppercase{\protect#1}}%
\def\rightmark{\@IEEEcompsoconly{\sffamily}\MakeUppercase{\protect#2}}}
\makeatother

\usepackage[amssymb]{SIunits}
\usepackage[cmex10]{amsmath} 
\usepackage{amssymb}
\usepackage{amsthm}
\usepackage[latin1]{inputenc}
\usepackage{cite}
\usepackage{url}
\usepackage{glossaries}
\usepackage{graphicx}
\usepackage[hang,small,bf]{caption}
\usepackage{subfig}
\usepackage[usenames]{color}
\usepackage{bm}
\usepackage{balance}

\newtheorem{lemma}{Lemma}
\newtheorem{theorem}{Theorem}

\newtheorem{corollary}{Corollary}
\newtheorem{remark}{Remark}

\newenvironment{DIFnomarkup}{}{}
\graphicspath{%
}%
\newacronym{MI}{MI}{mutual information} 
\newacronym{BER}{BER}{bit error rate} 
\newacronym{QAM}{QAM}{quadrature amplitude modulation} 
\newacronym{GMI}{GMI}{generalized mutual information} 
\newacronym{PDF}{PDF}{probability density function} 
\newacronym{PMF}{PMF}{probability mass function} 
\newacronym{AWGN}{AWGN}{additive white Gaussian noise} 
\newacronym{SNR}{SNR}{signal-to-noise ratio} 
\newacronym{ML}{ML}{maximum likelihood} 
\newacronym{BRGC}{BRGC}{binary reflected Gray code} 
\newacronym{NBC}{NBC}{natural binary code} 
\newacronym{BICM}{BICM}{bit-interleaved coded modulation} 
\newacronym{PAM}{PAM}{pulse-amplitude modulation} 
\newacronym{FBC}{FBC}{folded binary code} 
\newacronym{bpcu}{bpcu}{bits per channel use} 
\newacronym{iid}{i.i.d.}{independent identically distributed} 
\newacronym{iud}{i.u.d.}{independent uniformly distributed} 
\newacronym{MIMO}{MIMO}{multiple-input multiple-output} 
\newacronym{PSK}{PSK}{phase-shift keying} 
\newacronym{RV}{RV}{random variable}

\newcommand{\Pe}{p_{\mathrm{e}}}
\newcommand{\setC}{\mathcal{C}}

\newcommand{\eps}{\varepsilon}

\newcommand{\setS}{\mathcal{S}}

\newcommand{\SNR}{\rho}
\newcommand{\LL}{\boldsymbol{\mathbb{L}}}

\newcommand{\Expect}[1]{\mathrm{E}\left\{#1\right\}}
\newcommand{\MI}{\mathrm{MI}}

\renewcommand{\Pr}[1]{\mathrm{Pr}\left(#1\right)}

\newcommand{\GMI}{\mathrm{GMI}}
\newcommand{\BICMMI}{\mathrm{MI}^{\mathsf{bicm}}}
\newcommand{\BICMGMI}{\mathrm{GMI}^{\mathsf{bicm}}}
\newcommand{\HBICMGMI}{\mathrm{GMI}^{\mathsf{harm}}}

\newcommand{\e}{\mathrm{e}}
\newcommand{\sign}{\mathrm{sign}}

\newcommand{\pebar}{\bar{p}_{\e}}
\newcommand{\m}{\mathsf{m}}

\newcommand{\MyArg}[1]{_{#1}}

\newcommand{\BinSet}{\mathcal{B}}

\newcommand{\EqCo}{} 

\newcommand{\argmin}[2]{\mathop{\mathrm{argmin}}_{#1} {#2}}
\newcommand{\argmax}[2]{\mathop{\mathrm{argmax}}_{#1} {#2}}
\renewcommand{\max}[2]{\mathop{\mathrm{max}}_{#1} {#2}}

\newcommand{\eqsref}[2]{(\ref{#1})--(\ref{#2})}	
\newcommand{\figref}[1]{Fig.~\ref{#1}}

\newcommand{\theref}[1]{Theorem~\ref{#1}}
\newcommand{\secref}[1]{Section~\ref{#1}}

\newcommand{\bL}{\boldsymbol{L}}

\newcommand{\bLex}{\bL^{\mathsf{ex}}}
\newcommand{\bLml}{\bL^{\mathsf{ml}}}

\newcommand{\bLmix}{\tilde{\bL}}
\newcommand{\bLexmix}{\tilde{\bL}^{\mathsf{ex}}}
\newcommand{\bLmlmix}{\tilde{\bL}^{\mathsf{ml}}}

\newcommand{\Lex}{L^{\mathsf{ex}}}
\newcommand{\Lml}{L^{\mathsf{ml}}}

\newcommand{\lex}{l^{\mathsf{ex}}}
\newcommand{\lml}{l^{\mathsf{ml}}}

\newcommand{\Lmix}{\tilde{L}}
\newcommand{\Lexmix}{\tilde{L}^{\mathsf{ex}}}
\newcommand{\Lmlmix}{\tilde{L}^{\mathsf{ml}}}

\newcommand{\realR}{\mathbb{R}}

\begin{document}

\begin{DIFnomarkup}

\title{On the Information Loss of the Max-Log Approximation in BICM Systems}

\author{Mikhail~Ivanov, Christian Häger, Fredrik Br\"{a}nnstr\"{o}m, Alexandre Graell i Amat, Alex~Alvarado, and~Erik Agrell
\thanks{This research was supported by the Swedish Research Council, Sweden, under Grant No. 2011-5950, the Ericsson's Research Foundation, Sweden, and the European Community's Seventh's Framework Programme (FP7/2007-2013) under grant agreement No. 271986.}
\thanks{M. Ivanov, C. Häger, F. Br\"{a}nnstr\"{o}m, A. Graell i Amat, and E.~Agrell are with the Dept.~of Signals and Systems, Chalmers Univ.~of Technology, SE-41296 Gothenburg, Sweden (e-mail: \{mikhail.ivanov, christian.haeger, fredrik.brannstrom, alexandre.graell, agrell\}@chalmers.se).}
\thanks{A.~Alvarado is with the Optical Networks Group, Dept. of
Electronic \& Electrical Engineering, Univ. College London, WC1E
7JE London, UK (e-mail: alex.alvarado@ieee.org).}
}

\maketitle

\end{DIFnomarkup}

\begin{abstract}
 We present a comprehensive study of the information rate loss of the
 max-log approximation for $M$-ary pulse-amplitude modulation (PAM) in
 a bit-interleaved coded modulation (BICM) system. It is widely
 assumed that the calculation of L-values using the max-log
 approximation leads to an information loss. We prove that this
 assumption is correct for all $M$-PAM constellations and labelings
 with the exception of a symmetric 4-PAM constellation labeled
 with a Gray code. We also show that for max-log L-values, the BICM
 generalized mutual information (GMI), which is an achievable rate for
 a standard BICM decoder, is too pessimistic. In particular, it is  proved that the so-called ``harmonized'' GMI, which can be seen as the sum
 of bit-level GMIs, is achievable without any modifications to the
 decoder.  We then study how bit-level channel symmetrization and
 mixing affect the mutual information (MI) and the GMI for max-log
 L-values. Our results show that these operations, which are often
 used when analyzing BICM systems, preserve the GMI. However, this
 is not necessarily the case when the MI is considered. 
 Necessary and sufficient conditions under which these operations
 preserve the MI are provided. 
\end{abstract}

\begin{IEEEkeywords}
Bit-interleaved coded modulation, generalized mutual information, logarithmic likelihood ratio, max-log approximation, mismatched decoder.
\end{IEEEkeywords}

\glsresetall

\section{Introduction}\label{sec:intro}

\Gls{BICM} \cite{Zehavi92may,Fabregas08_Book,Alvarado15_Book} is a pragmatic approach to achieve high spectral efficiency
with binary error-correcting codes. Because of its inherent simplicity and
flexibility, as well as good performance, it is implemented in many
practical wireless communication
systems~\cite{IEEE80211-2012,ETSI_TS_136211-2013,
ETSI_EN_302_755_v131}.

A central part of a \gls{BICM} system is the demapper, which computes
soft information about the coded bits in the form of
so-called L-values. Ideally, L-values correspond to log-likelihood
ratios, in which case we refer to them as exact L-values. In practice,
however, the demapper often computes only approximate L-values due to
complexity reasons. A common approximation is to replace the log-sum
operation in the log-likelihood ratio computation with a max-log
operation.  This approximation can be motivated by the fact that at
high \gls{SNR}, exact and approximate max-log L-values are almost
identical. 

In this paper, we analyze achievable rates of \gls{BICM} for
nonbinary \gls{PAM} constellations over the
\gls{AWGN} channel, paying special attention to max-log L-values.
Traditionally, achievable rates for \gls{BICM} systems are analyzed
for exact L-values under the assumption of an ideal
interleaver~\cite{Caire98}, which results in the \gls{BICM} \gls{MI}
(i.e., the sum of $m$ bit-level MIs), often referred to as the
\gls{BICM} capacity. In~\cite{Martinez09}, it was proposed to analyze
 \gls{BICM} from a mismatched decoding perspective, showing
that the maximum achievable rate for a BICM system is lowerbounded by
the \gls{BICM} \gls{GMI}, without invoking any interleaver assumption.
For exact L-values, the \gls{BICM} \gls{GMI} coincides
with the BICM \gls{MI} \cite{Martinez09}.

When max-log L-values are considered, most of the previous work
concentrates on the correction of the ``suboptimal'' L-values in order
to either maximize the BICM \gls{GMI}~\cite{Jalden10, Nguyen11} or
minimize the error probability~\cite{Szczecinski12a}. To the best of
our knowledge, a rigorous comparison of achievable rates in terms of
the BICM MI and the BICM GMI for max-log L-values has not yet been
carried out.  Despite this fact, it seems to be a common belief in the
literature that the calculation of max-log L-values is inherently an
information lossy operation. As an example, when discussing the
\gls{MI} between the transmitted information symbol and the vector of
max-log L-values at the output of the demapper, \cite[p.
137]{Stierstorfer09_Thesis} concludes that ``the approximation clearly
constitutes a lossy procedure and entails an inferior BICM capacity''.
Similar implicit assumptions are made in~\cite{Alvarado07d}
and~\cite{Szczecinski07f}. We prove that this conclusion is not always true. In particular, we prove
that for symmetric $4$-\gls{PAM} constellations labeled with the
\gls{BRGC}, no information loss occurs when comparing exact and
max-log L-values, i.e., the BICM MI is the same in both cases. We also
prove that for all other combinations of \gls{PAM} constellations and
labelings, the max-log approximation indeed induces an information
loss. 

We then study the BICM \gls{GMI} for max-log L-values. In particular,
the so-called ``harmonized'' \gls{GMI} was introduced
in~\cite{Nguyen11} as an achievable rate for a modified BICM decoder
that applies scaling factors to the L-values. In this paper, we argue
that the L-value scaling is in fact unnecessary, and the harmonized
GMI (which can be seen as the sum of $m$ bit-level GMIs) is achievable
without any modifications to the decoder. Finally, we analyze two
common processing techniques which are often used in the theoretical
analysis of BICM systems: bit-level channel symmetrization and channel
mixing. The results show that these operations do not affect the BICM
GMI but can reduce the BICM MI.

The results presented in this paper can be easily generalized to
multi-dimensional product constellations of $M$-PAM of not necessarily
the same size 
labeled with a product labeling~\cite[Sec.~X]{Agrell04dec}.

\subsection{Notation}


Throughout the paper, boldface letters $\boldsymbol{x}$ denote row
vectors, blackboard letters $\mathbb{X}$ denote matrices, and capital
letters $X$ denote \glspl{RV}. $\mathbf{1}_{n}$ and $\mathbf{0}_{n}$
denote all-one and all-zero vectors of length $n$, respectively.
Calligraphic letters $\mathcal{X}$ denote sets, where $\realR$ stands
for the set of real numbers and $\mathbb{N}$ for the set of
natural numbers. For $n \in \mathbb{N}$, we define $[n] = \{1, 2, \dots,
n\}$. We define $\BinSet = \{0,1\}$. If $b \in \BinSet$, then
$\check{b} = (-1)^{(b+1)} \in \{-1, +1\}$ and $\bar{b} = 1 - b$.
$\Expect{\cdot}$ denotes expectation and $\Pr{\cdot}$ represents
probability. The \gls{PDF} of a continuous \gls{RV} $Y$ is denoted by
$f_{Y}(\cdot)$ and the conditional \gls{PDF} by $f_{Y|X}(\cdot |
\cdot)$. The \gls{PMF} of a discrete \gls{RV} $X$ is denoted by
$p_{X}(\cdot)$.

\section{System Model}

\begin{figure}
	\centering
	\includegraphics{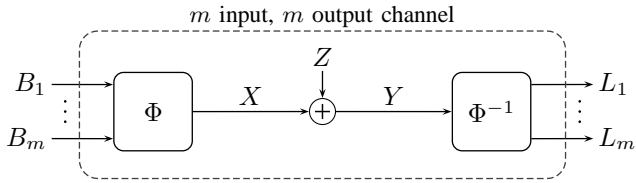}
	\caption{Block diagram of the analyzed system.}
	\label{fig:block_diag}
\end{figure}

A block diagram of the considered system model, which we discuss in
the following, is shown in Fig.~\ref{fig:block_diag}. 

\subsection{Modulator}\label{sec:modulator}

A modulator $\Phi$ is fed with $m$ bits $B_j$, $j \in [m]$, and
maps them to one of $M = 2^m$ possible constellation points. We
consider one-dimensional \gls{PAM} constellations denoted by
$\mathcal{S} = \{a_1,\dots, a_M\}$, where $a_1 < \dots < a_M$.
We say that the constellation is symmetric (around $y_0$) if
$a_{k} = -a_{M-k + 1} + 2y_0$ for $k \in [M]$ and some $y_0 \in
\mathbb{R}$, and we say that the constellation is equally spaced if
$a_{k+1} - a_{k}$ is independent of $k$. The bits are assumed to be
independent and distributed according to $p_{B_j}(u) =
1/2$,\,$\forall j$ and $u \in \mathcal{B}$. Thus, the symbols are
equiprobable, i.e., $p_{X}(a_k) = 1/M$,\, $\forall k  \in [M]$.
The constellation is assumed to be normalized to unit average energy
$\Expect{X^2}= (1/M)\sum_{k=1}^M{a_k^2} = 1$.

The mapping $\{0, 1\}^m \rightarrow \mathcal{S}$ performed by the modulator is assumed to
be one-to-one and is defined by a binary labeling. The binary labeling
is specified by an $m \times M$ binary matrix $\LL$, where the $k$th
column of $\LL$ is the binary label of the constellation point $a_k$.
Furthermore, we define $\setS_{j,u} = \{a_k \in \setS: \LL_{j,k} = u,
\forall k  \in [M]\}$ as the subconstellation consisting of all points labeled
with the bit $u$ in the $j$th bit position.

Certain quantities, such as the L-values we define below, depend only
on the subconstellations $\setS_{j,0}$ and $\setS_{j,1}$, i.e., they
depend only on the $j$th row in $\LL$. We refer to the $j$th row of
$\LL$ as a bit pattern, or simply pattern, which was shown in
\cite{Ivanov13a} to be a useful tool for analyzing binary labelings. A
pattern is defined as a vector $\boldsymbol{p}_j = [p_1, \dots, p_{M}]
\in \mathcal{B}^M$ with Hamming weight $M/2$. A labeling $\LL$ can then be
represented by $m$ different patterns, each corresponding to one row of $\LL$.
We define two trivial operations that can be applied to a pattern. A
\emph{reflection} of $\boldsymbol{p}$ is defined as $\boldsymbol{p}' =
\mathrm{refl}{(\boldsymbol{p})}$ with $p'_k = p_{M +1 - k}$. An
\emph{inversion} of $\boldsymbol{p}$ is defined as $\boldsymbol{p}' =
\mathrm{inv}{(\boldsymbol{p})}$ with $p'_k = \bar{p}_k$.  We say
that a pattern is symmetric if $\boldsymbol{p} =
\mathrm{refl}(\boldsymbol{p})$. A pattern $\boldsymbol{p}'$ that is
related to another pattern $\boldsymbol{p}$ via inversions and/or
reflections is said to be equivalent to $\boldsymbol{p}$. Analogously,
labelings related by trivial operations (i.e., row permutations,
inversion, and/or reflection of all patterns in the labeling) are said
to be equivalent~\cite[Definition~6b]{Agrell04dec}. For example, there
exist eight labelings that are equivalent to the \gls{BRGC} for
$4$-PAM shown in Fig.~\ref{fig:4-pam_brgc}(a). For symmetric
constellations, equivalent patterns and labelings behave similarly,
e.g., they give the same uncoded~\gls{BER} and achievable rates. 

Most of the numerical examples are presented for an equally spaced
4-PAM constellation, shown in
Fig.~\ref{fig:4-pam_brgc}\subref{fig:subplot_4pam1}--\subref{fig:subplot_4pam2},
together with the two labelings
\begin{align}
	\LL_1 = \begin{bmatrix}
		0 & 0 & 1 & 1 \\
		0 & 1 & 1 & 0 
	\end{bmatrix}, \qquad
	\LL_2 = \begin{bmatrix}
		0 & 0 & 1 & 1 \\
		0 & 1 & 0 & 1 
	\end{bmatrix}\EqCo
\end{align}
which are often referred to as the \gls{BRGC} and the \gls{NBC} or
set-partitioning labeling~\cite{Unger82jan}, respectively. An example
of an equally spaced $8$-PAM constellation labeled with the \gls{BRGC}
is shown in~Fig.~\ref{fig:4-pam_brgc}\subref{fig:subplot_8pam}.

\begin{figure}
	\centering
	\subfloat[$4$-PAM with BRGC.]{\label{fig:subplot_4pam1}
	\includegraphics{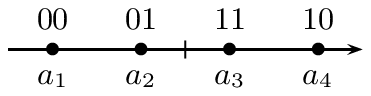}
	}
	\qquad
	\subfloat[$4$-PAM with NBC.]{\label{fig:subplot_4pam2}
	\includegraphics{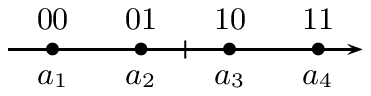}
	}
	
	\subfloat[$8$-PAM with BRGC.]{\label{fig:subplot_8pam}
	\includegraphics{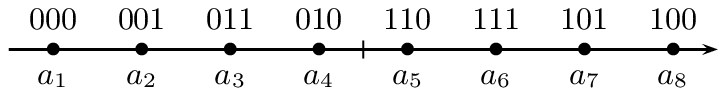}
	}
	\caption{Examples of equally spaced PAM constellations with different binary labelings.}
	\label{fig:4-pam_brgc}
\end{figure}

\subsection{AWGN Channel}

The constellation points are assumed to be transmitted over the discrete-time
memoryless \gls{AWGN} channel with output $Y = X + Z$, where the noise
$Z$ is a zero-mean Gaussian \gls{RV} with variance
$\Expect{Z^2} = N_0/2$. The conditional \gls{PDF} of the channel
output given the channel input is
\begin{equation}
    f_{Y|X}(y|x) = \sqrt{\frac{\SNR}{\pi}} e^{-\SNR(y -
		x)^2}\EqCo
    \label{eq:gauss}
\end{equation}
where $\SNR = {1}/{N_0}$ is the average \gls{SNR}.

\subsection{Demappers and L-values}
Two demappers $\Phi^{-1}$ are considered at the receiver. The first one
calculates \emph{exact} L-values as the log-likelihood
ratios
\begin{equation}
	\lex_j(y) = \log \frac{f_{Y|B_j}(y | 1)}{f_{Y|B_j}(y | 0)} = 
	\log{ \frac{\sum_{x \in \setS_{j,
	1}}{e^{-\SNR(y-x)^2}}}{\sum_{x \in \setS_{j,
	0}}{e^{-\SNR(y-x)^2}}}}.
    \label{eq:Lvalue_exact}
\end{equation}
The second demapper calculates \emph{max-log}
L-values using the max-log approximation as~\cite{Viterbi98feb}
\begin{equation}
	\lml_j(y) = \SNR\left[\min_{x \in \setS_{j,0}}{(y-x)^2} - \min_{x \in \setS_{j, 1}}{(y-x)^2}\right].
    \label{eq:Lvalue_maxlog}
\end{equation} 
The observation $Y$ is an \gls{RV} and thus, so are the
L-values. To simplify the notation, we use $L_j = \Lex_j = \lex_j(Y)$
when discussing exact L-values and $L_j = \Lml_j = \lml_j(Y)$ when
discussing max-log L-values. We further define the vector $\bL = [L_1,
\dots, L_m]$ and write $\bLex$ and $\bLml$ when discussing exact and
max-log L-values, respectively. 

\begin{figure}[t]
	\centering
	\subfloat[$\rho = 0$ dB.]{\includegraphics{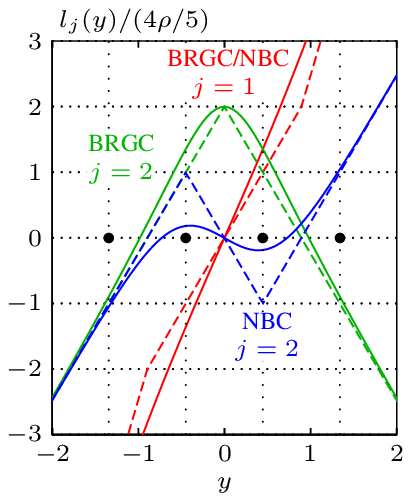}}
	\quad
	\subfloat[$\rho = 6$ dB.]{\includegraphics{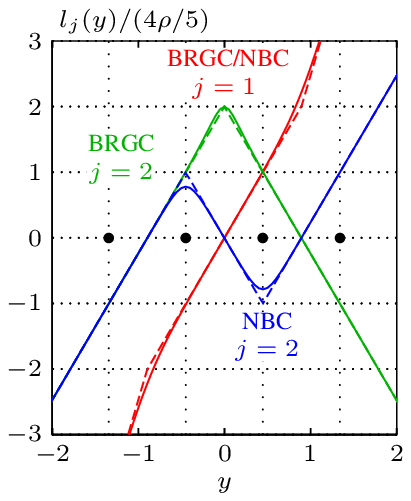}}

	\caption{Normalized exact (solid) and max-log (dashed) L-values as
	functions of the observation $y$ assuming an equally spaced 4-PAM constellation
	labeled with the BRGC and NBC (see Fig.~\ref{fig:4-pam_brgc}(a)--(b)).}
	\label{fig:lvalues}
\end{figure}

In Fig.~\ref{fig:lvalues}, we show an example of the exact and max-log
L-values (normalized by $4\rho/5$) as functions of the observation $y$
for the $4$-PAM constellation and labelings shown in
Fig.~\ref{fig:4-pam_brgc}\subref{fig:subplot_4pam1}--\subref{fig:subplot_4pam2}
and two different values of $\rho$. As shown in~\cite{Alvarado07d},
the max-log L-value is a piecewise linear function of the observation,
which simply scales with \gls{SNR}, whereas the dependency of the
exact L-value on the \gls{SNR} is nonlinear. However, when the
\gls{SNR} increases, one can show that the exact L-value approaches
the max-log L-value, in the sense that $\lim_{\rho \to \infty}
\lex_j(y)/\rho = \lml_j(y)/\rho$ $ \forall y \in
	\mathbb{R}$, where $\lml_j(y)/\rho$ is
independent of $\rho$ (see \eqref{eq:Lvalue_maxlog}).

From Fig.~\ref{fig:lvalues}, one can observe that the exact
L-value for the second bit position of the BRGC (and also the max-log
L-value) is an even function assuming an equally spaced 4-PAM
constellation. More generally, we have the following result, which
will be used later on.

\begin{lemma}
	\label{lemma:exact_symmetric}
	The exact L-value $\lex_j(y)$ is symmetric with respect to $y_0 \in
	\mathbb{R}$, i.e., $\lex_j(y_0+y) = \lex_j(y_0-y)$ for $y \in
	\mathbb{R}$, if and only if the constellation is symmetric around
	$y_0$ and the pattern corresponding to the $j$th bit position
	satisfies $\boldsymbol{p} = \mathrm{refl}(\boldsymbol{p})$. 
\end{lemma}
\begin{proof}
	The proof is given in~Appendix~\ref{Appendix.theor:exact_symmetric}.
\end{proof}

\begin{remark}
	\label{remark:exact_symmetric_unique}
	Since the exact L-value $\lex(y)$ is not a periodic function, the
	symmetry point is unique, i.e., there cannot exist two distinct
	$y_0, y_0' \in \mathbb{R}$ such that both $\lex(y+y_0) =
	\lex(-y+y_0)$ and $\lex(y+y_0') = \lex(-y+y_0')$ hold for all $y \in
	\mathbb{R}$.  
\end{remark}

\begin{remark}
	\label{remark:maxlog_symmetric}
	It can be shown that Lemma \ref{lemma:exact_symmetric} holds without
	change also for the max-log L-value.
\end{remark}

\subsection{Coding Scheme}\label{sec:coding_general}

We consider a coding scheme where an information message is mapped to
a codeword $\boldsymbol{x} = [x_1, \dots, x_N]$, $x_i \in \setS$ for
$i \in [N]$, and $N$ corresponds to the number of (AWGN) channel uses.
The set of all possible codewords $\boldsymbol{x}$ is a nonbinary code
of length $N$. As the mapping $\Phi$ is one-to-one, an alternative
binary code $\setC$ of length $mN$ can be constructed. At
the transmitter, a binary codeword $\boldsymbol{b} \in \setC$ is
selected and at the receiver a length-$m N$ vector of L-values
$\boldsymbol{l}$ is calculated. With a slight abuse of notation, we
write $b_{j,i}$ and $l_{j,i}$ to denote the $j$th input bit to the
modulator and the $j$th L-value from the demapper in the $i$th channel
use, respectively. In the length-$mN$ vector $\boldsymbol{b}$ (and
similarly for $\boldsymbol{l}$) $b_{j,i}$ corresponds to the entry
$b_{N(j-1)+i}$. This way, all input bits that correspond to a
particular bit position appear consecutively in $\boldsymbol{b}$,
i.e., $\boldsymbol{b} = [\dots, b_{j,1}, b_{j,2},\dots ,b_{j,N},
b_{j+1,1},\dots]$.

The standard \gls{BICM} decoder we consider in this paper is defined
as~\cite[eq.~(3)]{Martinez09}
\begin{equation}
\label{eq:bicm_decoder}
\hat{\boldsymbol{b}} = \argmax{\boldsymbol{b} \in \setC}{\sum_{i=1}^{N}
\sum_{j=1}^{m} b_{j,i} l_{j,i}} \EqCo
\end{equation}
i.e., the decoder finds the codeword that maximizes the correlation
with the vector of the observed L-values. The codeword error
probability is defined as $\Pe = \mathrm{Pr}(\hat{\boldsymbol{B}}
\neq \boldsymbol{B})$. 

To simplify the notation, one of the indices $i$, $j$ may be omitted
depending on the discussed context. To avoid ambiguity, in the
rest of the paper, the following convention applies: the index $j
\in [m]$ is used to indicate the bit position and the index $i \in
[N]$ is used to indicate the time instant.

\section{Bit-Level Analysis}
\label{sec:patterns}

\begin{figure}[t]
	\centering
	\begin{center}
		\includegraphics{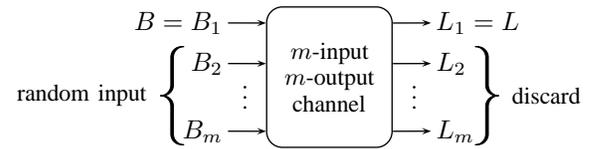}
	\end{center}
	\caption{Bit-level channel, illustrated for the first bit position.}
	\label{fig:bit_level}
\end{figure}

When analyzing achievable rates of \gls{BICM}, it is common to proceed
with a parallel independent channel model and assume that there exist
$m$ independent bit channels from $B_j$ to $L_j$. This assumption is
typically motivated by the insertion of the so-called ``ideal
interleaver''~\cite[Sec.~II-B]{Caire98}. In this paper, we use a
different approach which does not rely on any interleaver or
independence assumption between bit channels. We reduce the $m$-input
$m$-output channel in Fig.~\ref{fig:block_diag} to a channel with only
one binary input and one continuous output. This can be done by
specifying a behavioral model for the other, unused bit positions. To
that end, consider the hypothetical scenario where we are only
interested in transmitting data from $B_j$ to $L_j$. To do so, we may
feed the modulator at all other bit positions $j' \neq j$ with
\gls{iud} bits. If the \gls{iud} condition is not satisfied, the
symbols are not equiprobable and the system model assumptions
in~\secref{sec:modulator} are violated. At the receiver side, we
discard all L-values except the one of interest. This is conceptually
shown in Fig.~\ref{fig:bit_level} for $j=1$ leading to a binary-input,
continuous-output channel from $B_1$ to $L_1$. 

Since the results in this section are not dependent on any particular
$j$, the bit-level index $j$ is dropped. Definitions and equations
that hold for both exact and max-log L-values will be stated with the
generic placeholder variable $L$. As an example, the generic bit-level
channel in~\figref{fig:bit_level} is denoted by
$f_{L|B}(\cdot|\cdot)$.  For exact L-values, the channel is then
denoted by $f_{\Lex|B}(\cdot|\cdot)$. Note that this conditional PDF
is hard to calculate in general \cite[Ch.~4]{Fabregas08_Book}. For max-log
L-values, the channel is denoted by $f_{\Lml|B}(\cdot|\cdot)$. This
conditional PDF is relatively easy to obtain due to the special form
of~\eqref{eq:Lvalue_maxlog} and corresponds to a summation
of piecewise Gaussian functions (see \cite{Alvarado07d} and references therein).

\subsection{Bit-Level Coding Scheme} \label{sec:bit_level_coding}

The coding scheme for the bit-level channel in~\figref{fig:bit_level}
is obtained from the one described in~\secref{sec:coding_general} by simply
omitting irrelevant bit positions. Let $\setC \subset \mathcal{B}^N$
denote a binary code of length $N$ and rate $R =
\log_2(|\setC|)/N$. The decoder in~\eqref{eq:bicm_decoder} then
reduces to 
\begin{align}
	\boldsymbol{\hat{b}} = \argmax{\boldsymbol{b} \in \setC}{
	\sum\limits_{i=1}^N b_i l_i}\EqCo
	\label{eq:decoder_general}
\end{align}
where the $l_i$ are either exact or max-log L-values.

\subsection{Achievable Rates}

\subsubsection{Generalized Mutual Information}

The maximum achievable rate for the decoder
in~\eqref{eq:decoder_general} is lowerbounded by
\cite[eq.~(25)]{Ganti00}\footnote{To obtain \eqref{eq:gmi_general}
from \cite[eq.~(25)]{Ganti00}, the decoding metric in \cite{Ganti00}
is chosen as $d(\check{B},L) = -\check{B}L$ together with $a(\check{B}) = 0$, where the minus sign comes from the fact that the metric in \cite{Ganti00} is minimized, whereas it is maximized in~\eqref{eq:decoder_general}.}
(see also~\cite[eq.~(18)]{Nguyen11})
\begin{align}
	\GMI_{L} = 1 - \inf_{s\geq0} \Expect{\log_2 \left(1 + e^{-s\check{B} L}\right)} \EqCo
	\label{eq:gmi_general}
\end{align}
which was originally introduced in~\cite{Merhav1994} for discrete memoryless channels.
We refer to this quantity as the bit-level GMI or simply GMI. 
 The \gls{GMI} has the following operational meaning. There exists a binary
code $\setC$ with rate arbitrarily close to $\GMI_{L}$ that can
achieve reliable communication (i.e., $\Pe < \eps$ for $\eps$ as small
as desired) as $N\to \infty$ over the channel from $B$ to $L$
\emph{assuming the decoder in~\eqref{eq:decoder_general}}.  The
codewords of such a code $\setC$ are composed of \gls{iud}
bits~\cite{Ganti00} and such codes are called \gls{iud} codes.

\subsubsection{Mutual Information}

Lifting the decoder assumption, the largest achievable rate for the channel
in~\figref{fig:bit_level} is given by the \gls{MI} between $B$ and $L$
defined as \cite[p.~251]{Cover2006second}
\begin{equation}\label{eq:bit_level_mi}
	\MI\MyArg{L} = I(B;L) = \Expect{\log_2
	\left(\frac{f_{L|B}(L|B)}{f_{L}(L)} \right)}.
\end{equation}
The \gls{MI} has a similar operational meaning as the \gls{GMI}, but
does not make any restrictions regarding the decoder structure.  In
particular, for the channel from $B$ to $L$, there exists a binary
code $\setC$ with rate arbitrarily close to $\MI_L$ that can
achieve reliable communication as $N \to \infty$. Furthermore, the
\gls{MI} is the \emph{maximum} achievable rate.  Note that both the
GMI and the MI are functions of the \gls{SNR}.  However, to simplify
the notation, we omitted the dependence on $\rho$. 

\subsection{L-values}

\subsubsection{Exact L-values}

For exact L-values, the decoder~\eqref{eq:decoder_general} corresponds
to the maximum-likelihood decoder for the channel $f_{\Lex |B}(\cdot
|\cdot)$. This explains that for exact L-values, the \gls{GMI} is
equivalent to the \gls{MI}. In fact, the \gls{MI} for exact L-values
can alternatively be written as
\begin{align}
	\MI_{\Lex} &= I(B; \Lex) = I(B;Y)  = \GMI_{\Lex}
	\label{eq:mi_ex}
\end{align}
and the infimum in~\eqref{eq:gmi_general} is achieved for $s = 1$, as
shown in~\cite[Cor.~1]{Martinez09}.

\subsubsection{Max-Log L-values}

For max-log L-values, the \gls{MI} is given by
\begin{equation}
	\MI_{\Lml} = I(B; \Lml).
\end{equation}
Unlike for exact L-values, $\MI_{\Lml} \ge \GMI_{\Lml}$. This is
because max-log L-values are not true log-likelihood ratios for the
channel $f_{\Lml|B}(\cdot|\cdot)$ and hence, the decoder
in~\eqref{eq:decoder_general} does not correspond to a maximum-likelihood decoder. However, applying different functions to $\Lml$
may increase the corresponding \gls{GMI}, which is in sharp contrast
to the \gls{MI} and the data processing inequality \cite[Th.~2.8.1]{Cover2006second}. In~\cite[Th.~1]{Jalden10} (see also \cite[Th.~7.5]{Alvarado15_Book}), it is shown that
$\MI_{\Lml} = \GMI_{g(\Lml)}$ for 
\begin{align}\label{eq:optimal_correction}
	g(l)  
	& = \log\left(\frac{f_{\Lml|B}(l|1)}{f_{\Lml|B}(l|0)}\right).
\end{align}
The intuitive interpretation is that the processing
in~\eqref{eq:optimal_correction} matches the metrics to the
decoder~\eqref{eq:decoder_general}, and hence, makes the decoder a
maximum-likelihood decoder.

We can compare the discussed achievable rates in the form of the
following chain of inequalities 
\begin{align}
\label{ineq.pattern}
\GMI_{\Lex} = \MI_{\Lex} \overset{(a)}{\geq} \MI_{\Lml} =
\GMI_{g(\Lml)} \geq \GMI_{\Lml}\EqCo
\end{align}
where inequality (a) follows from the data processing inequality. As
mentioned in~\secref{sec:intro}, it is commonly assumed that
inequality (a) is strict. In the next section, we show that this
inequality is in fact an equality in some cases.

\subsection{Lossless Max-Log Approximation}\label{sec:lossless_patterns}

We start with the following lemma. 

\begin{lemma}\label{lemma:min_suf_stat}
	For any one-dimensional constellation and any pattern, $\MI_{\Lex} =
	\MI_{\Lml}$ if and only if there exists a function $f(\cdot)$ such
	that $\Lex = f(\Lml)$.
\end{lemma}
\begin{proof}
	The ``if'' part follows from the data processing inequality. The
	``only if'' follows from the fact that exact L-values form a
	\emph{minimal} sufficient statistic for guessing $B$ based on $Y$.
	A minimal sufficient statistic is a function of every other
	sufficient statistic. In particular, assume $\lml(y) = \lml(y')$
	for two channel observations $y$ and $y'$. If $\MI_{\Lml} =
	\MI_{\Lex}$ (and hence max-log L-values also form a sufficient
	statistic), it follows from Fisher's factorization
	theorem~\cite[Ch.~22.3]{Lapidoth2009} that $f_{Y|B}(y|b) /
	f_{Y|B}(y'|b)$ is independent of $b$, which implies $\lex(y) =
	\lex(y')$. Thus, there has to exist a function $f(\cdot)$ such that
	$\Lex = f(\Lml)$. 	
\end{proof}

Based on this lemma, we have the following theorem.

\begin{theorem}\label{theor:lossless_patterns}
	For one-dimensional $M$-PAM constellations, there exist only two
	cases for which the max-log approximation is information lossless,
	i.e., $\MI_{\Lex} = \MI_{\Lml}$. Either the pattern is equivalent to
	$\boldsymbol{p}_{\mathrm{I}} = [\mathbf{0}_{M/2},
	\mathbf{1}_{M/2}]$, in which case the constellation can be
	arbitrary, or the pattern is equivalent to
	$\boldsymbol{p}_{\mathrm{II}} = [\mathbf{0}_{M/4}, \mathbf{1}_{M/2},
	\mathbf{0}_{M/4}]$ and the constellation is symmetric. 
\end{theorem}

\begin{proof}
	The proof is given in~Appendix~\ref{Appendix.theor:lossy_patterns}.
\end{proof}

In practice, \theref{theor:lossless_patterns} implies that, for the
two lossless cases, ``full'' information can be extracted from max-log
L-values if proper processing is applied, i.e., in the form of the correction function~\eqref{eq:optimal_correction}. In addition to the patterns
in~\theref{theor:lossless_patterns}, we denote the pattern
$\boldsymbol{p} = [0, 1, 0, 1, \dots, 0, 1]$ by
$\boldsymbol{p}_{\mathrm{III}}$. As an example, for $4$-PAM, there
exist only three patterns that are not equivalent, i.e.,
$\boldsymbol{p}_{\mathrm{I}} = [0, 0, 1, 1]$,
$\boldsymbol{p}_{\mathrm{II}}= [0, 1, 1, 0]$, and
$\boldsymbol{p}_{\mathrm{III}} = [0, 1, 0, 1]$. For a symmetric
$4$-PAM constellation, the first two patterns are lossless according
to~\theref{theor:lossless_patterns} and they correspond to the first
and the second bit position in the BRGC, respectively. For the
\gls{NBC}, the first and second bit positions correspond to the
patterns $\boldsymbol{p}_{\mathrm{I}}$ and
$\boldsymbol{p}_{\mathrm{III}}$, respectively (see
Fig.~\ref{fig:4-pam_brgc}). From~\theref{theor:lossless_patterns}, the
first bit position is again information lossless (even if the
constellation is not symmetric) while the second one is not. In fact,
we immediately have the following corollary.

\begin{corollary}
	\label{theorem:lossless_labelings}
	Among all possible combinations of one-dimensional $M$-PAM
	constellations and labelings, a symmetric $4$-PAM constellation
	with the \gls{BRGC} (or any equivalent labeling) is the only case
	where all bit positions are information lossless. 
\end{corollary}

\begin{proof}
	It is easy to show that $\boldsymbol{p}_{\mathrm{I}}$ and
	$\boldsymbol{p}_{\mathrm{II}}$ or their equivalent patterns cannot be used twice in a labeling. This means that
	any labeling with more than two bit positions will contain a bit
	pattern for which the max-log approximation causes an information loss.
\end{proof}
It is interesting to look at the
function~\eqref{eq:optimal_correction} and compare it with the curve
obtained by plotting $\lex$ versus $\lml$, as shown
in~\figref{fig:maxlog_vs_exact} for the three non-equivalent patterns
$\boldsymbol{p}_{\mathrm{I}}$ (red),	$\boldsymbol{p}_{\mathrm{II}}$
(green), and $\boldsymbol{p}_{\mathrm{III}}$ (blue). In general, for
the lossless patterns, this function coincides with the curve $\lex$
versus $\lml$ and for lossy patterns it does not. The information loss
then comes from the region where $g(\cdot)$ cannot recover the exact
L-value. 

\begin{figure}
	\centering
	\subfloat[$\SNR = 0$~dB.]{
	\includegraphics{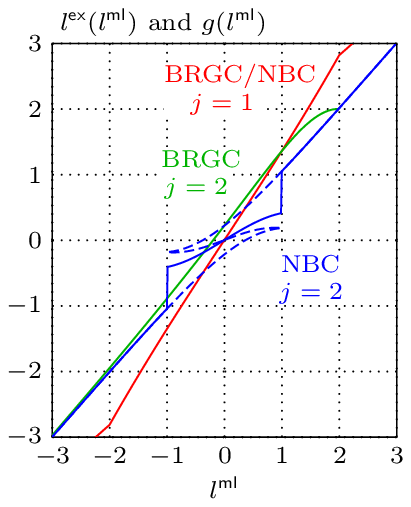}
	}
	\subfloat[$\SNR = 6$~dB.]{
	\includegraphics{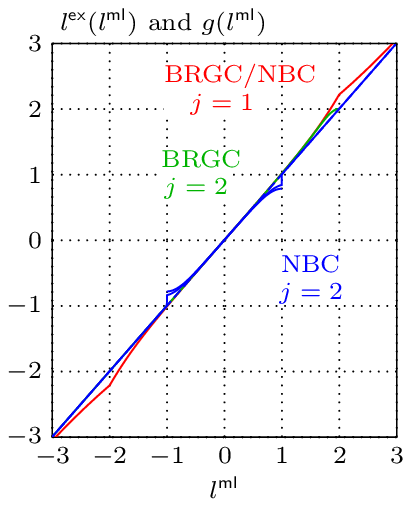}
	}

	\caption{The correction function $g(\cdot)$ (solid) and the exact
	L-value versus the max-log L-value (dashed) for the three
	non-equivalent patterns for an equally spaced 4-PAM constellation. The values on the x- and y-axes are normalized by $4\SNR/5$.}
	\label{fig:maxlog_vs_exact}
\end{figure}

\section{BICM Analysis} 
\label{sec:labelings}

In this section, we return from the bit-level viewpoint to the
original $m$-input $m$-output channel shown in
Fig.~\ref{fig:block_diag}. 

\subsection{BICM Mutual Information}

The \gls{BICM} \gls{MI} is defined as 
\begin{align}
	\label{eq:bicm_capacity}
	\BICMMI_{\bL} = \sum_{j=1}^{m} \MI_{L_j} = \sum_{j=1}^{m}
	I(B_j; L_j)\EqCo
\end{align}
i.e., the sum of $m$ bit-level \glspl{MI}, for both exact
(cf.~\cite[eq. (15)]{Caire98}) and max-log L-values.  Under the
parallel independent channel model assumption~\cite{Caire98}, it is
the maximum achievable rate for exact L-values with the standard
\gls{BICM} decoder~\eqref{eq:bicm_decoder}.  However, in the case of
the model in~\figref{fig:block_diag}, its operational meaning as an
upper bound on the achievable rate is unclear. Using the mismatched
decoding framework, it was shown to be an achievable rate for the
standard BICM decoder \cite[Sec.~III]{Martinez09}. It is also
achievable for max-log L-values, provided that the ideal correction
function $g(\cdot)$ is applied for each bit level before decoding via
\eqref{eq:bicm_decoder}.

\subsection{BICM Generalized Mutual Information}

The BICM GMI for \gls{iud} input bits can be written as \cite[eq.~(62)]{Martinez09}
\begin{equation}
	\label{eq:bicm_gmi}
	\BICMGMI_{\bL} = m - \inf_{s\geq0}{ \sum_{j = 1}^{m} \Expect{\log_2
	\left(1 + e^{-s \check{B}_j L_j}\right)} }
\end{equation}
and was shown to be an achievable rate for the decoder
in~\eqref{eq:bicm_decoder}~\cite[Sec.~III]{Martinez09}.  For exact
L-values, similarly to the bit-level GMI in~\eqref{eq:gmi_general}
(see also~\eqref{eq:mi_ex}), the value $s = 1$ maximizes the BICM GMI
in~\eqref{eq:bicm_gmi}. In that case, the BICM GMI can be written as a
sum of bit-level GMIs. This, however, does not hold for max-log
L-values. 

It has recently been shown in~\cite{Nguyen11} that the so-called
``harmonized'' \gls{GMI} defined as \cite[eqs.~(18),
(21)]{Nguyen11}
\begin{multline}
\label{eq:bicm_s_gmi}
	\HBICMGMI_{\bL} = \sum_{j = 1}^{m} \GMI_{L_j} \\= m - \sum_{j =
	1}^{m} \inf_{s_j \geq 0} \Expect{\log_2 \left( 1 + e^{-s_j\check{B}_j L_j} \right)}
\end{multline}
is achievable when max-log L-values are used and different linear
corrections are applied to L-values at different bit levels. Note
that, unlike the BICM GMI, the harmonized GMI in~\eqref{eq:bicm_s_gmi}
corresponds to the sum of $m$ bit-level GMIs for both exact \emph{and}
max-log L-values. We show in the following theorem that the
harmonized GMI is achievable by the standard BICM decoder without the
assumption of any L-value correction. 
\begin{theorem}
	\label{theorem:harmonized_gmi}
	For any one-dimensional constellation and any labeling, the rate $\HBICMGMI_{\bL}$ in~\eqref{eq:bicm_s_gmi} is achievable by the standard BICM
	decoder~\eqref{eq:bicm_decoder}.  
\end{theorem}

\begin{proof}
	The proof is given in~Appendix~\ref{Appendix.theor:harmonized_gmi}.
\end{proof}

\begin{remark}
	The BICM GMI \eqref{eq:bicm_gmi} is the largest rate for which
	the average error probability, averaged over all messages and
	\gls{iud}~codes, vanishes. On the other hand, the proof of Theorem
	\ref{theorem:harmonized_gmi} relies on codes that are constructed as
	the Cartesian product of $m$ \gls{iud} codes and, hence, the
	overall code is not \gls{iud} The fact that rates larger than those
	given by the \gls{GMI} can be achieved with non-\gls{iud} codes and
	mismatched decoders has also been observed in~\cite{Lapidoth96}.
\end{remark}

\subsection{Inequalities}

To clarify the difference between the achievable rates for exact and
max-log L-values discussed in this section, we give a short summary in
the form of the following inequalities. For exact L-values, all
studied quantities are the same and we have
\begin{align}
\label{ineq.labeling.exact}
& \BICMMI_{\bLex} = \HBICMGMI_{\bLex} = \BICMGMI_{\bLex}\EqCo
\end{align}
which is the rate that is achievable by the standard BICM decoder.

The value in~\eqref{ineq.labeling.exact} is an upper bound on the
corresponding quantities for max-log L-values, i.e.,
\begin{align}
\label{ineq.labeling.maxlog}
\BICMMI_{\bLex} \geq \BICMMI_{\bLml} \geq \HBICMGMI_{\bLml} \geq \BICMGMI_{\bLml}.
\end{align}
As previously mentioned, the second quantity is an achievable rate if
the function~\eqref{eq:optimal_correction} is applied to the max-log
L-values from the $m$ different bit positions before passing them to
the decoder~\eqref{eq:bicm_decoder}. The third quantity is a rate
achievable by the standard BICM decoder~\eqref{eq:bicm_decoder}
without any L-value correction. The last quantity corresponds to the
BICM GMI as defined in \cite[eq.~(59)]{Martinez09} for max-log
L-values. As shown in~Corollary~\ref{theorem:lossless_labelings},
for one-dimensional constellations the first inequality is an equality
only for a symmetric $4$-PAM constellation labeled with a binary
labeling equivalent to the BRGC.

\section{L-value Processing}

\subsection{Symmetrization} 

In this subsection, we study how bit-level channel symmetrization
affects the GMI and the MI for exact and max-log L-values. Bit-level
symmetrization can be motivated as follows. A binary input channel
$f_{L|B}(\cdot|\cdot)$ is said to be output-symmetric if 
\begin{align}
	f_{L|B}(l|b) = f_{L|B}(-l|\bar{b}).
\end{align}
For some patterns, for instance for $\boldsymbol{p} = [0, 1, 1, 0]$,
the channel is not output-symmetric (neither for exact nor max-log
L-values), which complicates the analysis of BICM systems in certain
cases. For example, one cannot assume the transmission of the all-zero
codeword when studying the error probability $\Pe$ of a linear code
over such a channel. 

\begin{figure}[t]
	\centering
	\begin{center}
		\includegraphics{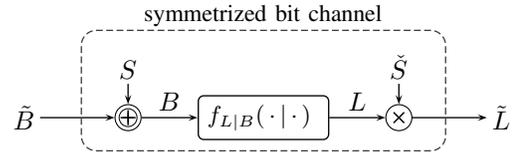}
	\end{center}
	\caption{Illustration of the channel symmetrization technique via
	i.u.d. bit adapters. $S$ is assumed to be known by both the
	transmitter and receiver and added to $B$ modulo 2 at the
	transmitter. }
	\label{fig:symmetrization}
\end{figure}
\begin{figure}[t]
	\centering
	\begin{center}
		\includegraphics{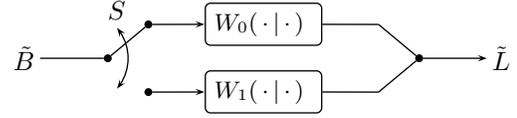}
	\end{center}
	\caption{Equivalent model for the channel symmetrization where $S$
	is interpreted as a switch and we have $W_0(l|b) = f_{L|B}(l|b)$ (for
	$S = 0$) and
	$W_1(l|b) = f_{L|B}(-l|\bar{b})$ (for $S = 1$). }
	\label{fig:symmetrization_state}
\end{figure}

To enforce an output symmetric channel, it was proposed
in~\cite{Caire98} to use a randomly complemented labeling.
In~\cite{Hou2003}, a similar symmetrization technique was realized by
the use of random \gls{iid} bit adapters as shown in
Fig.~\ref{fig:symmetrization}. At the transmitter, a uniformly random
bit $S$ is added modulo 2 to the transmitted bit and at the receiver,
the L-value is multiplied by $\check{S} = (-1)^S$. The system can be thought of in the following way.
The value of the adapter $S$ is known at both the transmitter and the receiver side, however, it is
not known to the encoder and decoder.\footnote{The considered system model
is equivalent to the one in~\secref{sec:patterns} if $S$ is known to
the decoder.} Hence, the adapter can be
considered part of the channel. If we denote the symmetrized L-value by $\tilde{L}$, the
conditional \gls{PDF} of $\tilde{L}$ can be related to the conditional
\gls{PDF} of the original L-value through
\begin{equation}
	f_{\tilde{L}|\tilde{B}}(l|b) = \frac{1}{2}\left(f_{L|B}(l|b) +
f_{L|B}(-l|\bar{b})\right).
\label{eq:symmetric_pdf}
\end{equation}
In~\cite[Th.~2]{Hou2003}, it was shown that the MI is unchanged by
the symmetrization if exact L-values are used. Somewhat surprisingly,
the effect of this operation on the MI and the GMI for max-log
L-values has not been studied in the literature. In the following, we show that the
\gls{GMI} is not affected by the symmetrization, while the \gls{MI} is
reduced for max-log L-values.

\begin{theorem} \label{theo:gmi_symmetrization}
	For any one-dimensional constellation and any pattern, the bit-level channel symmetrization does not change the \gls{GMI}
	in~\eqref{eq:gmi_general}, i.e., $\GMI_{\tilde{L}} = \GMI_{L}$.
\end{theorem}
\begin{proof}
	Let $h(x) =	\log_2(1+e^{-x})$. The expectation in~\eqref{eq:gmi_general} with respect to
	$\tilde{L}$ can then be written as
	\begin{align}
		&\Expect{h(s\check{\tilde{B}}\tilde{L})} \\ & = \frac{1}{2}\sum_{b \in
		\BinSet}\int_{-\infty}^{\infty}f_{\tilde{L}|\tilde{B}}(l|b)
		h(s\check{b}l)\, \mathrm{d}l\\ 
		& \stackrel{\eqref{eq:symmetric_pdf}}{=} \frac{1}{4}\sum_{b \in
		\BinSet}\int_{-\infty}^{\infty}\left(f_{L|B}(l|b) +
		f_{L|B}(-l|\bar{b})\right) h(s\check{b}l)\, \mathrm{d}l\\ 
		&=
		\frac{1}{2}\sum_{b \in
		\BinSet}\int_{-\infty}^{\infty}f_{L|B}(l|b) h(s\check{b}l)\,
		\mathrm{d}l\\ & = \Expect{h(s\check{B} L)}.
	\end{align}
\end{proof}

Intuitively, this result can be explained by the fact that the
decoder~\eqref{eq:decoder_general} does not exploit the information
about the asymmetry of the L-values even if it is available. 

The effect of the channel symmetrization on the mutual information
is described in the following corollary.

\begin{corollary}\label{corol:symmetr}
	For any one-dimensional constellation and any pattern, the bit-level channel symmetrization does not change the \gls{MI}, i.e.,
	$\MI_{\Lexmix} = \MI_{\Lex}$. Furthermore, we have $\MI_{\Lmlmix}
	\leq \MI_{\Lml}$ with equality if and only if the correction
	function in \eqref{eq:optimal_correction} is odd, i.e., $g(l) = -
	g(-l)$. 
\end{corollary}
\begin{proof}
	The corollary follows from Theorem
	\ref{theorem:mutual_information_convexity} in Appendix
	\ref{Appendix.theor:mutual_information_convexity}. Indeed, the
	scrambler $S$ can be thought of as a switch between the two different
	channel laws $W_0(l | b) = f_{L|B}(l|b)$ for $S = 0$ and $W_1(l | b)
	= f_{L|B}(-l|\bar{b})$ for $S = 1$, see
	Fig.~\ref{fig:symmetrization_state}. Observe that we have $\MI_{L} =
	I(B; L) = I(\tilde{B}; \tilde{L}|S) \geq I(\tilde{B};\tilde{L}) =
	\MI_{\tilde{L}}$. The necessary and sufficient condition for
	equality according to Theorem
	\ref{theorem:mutual_information_convexity} is $g_0(l) = g_1(l)$,
	where $g_j(l)$ is defined in \eqref{eq:convexity_correction}.
	This condition can be written as $g_0(l) = -g_0(-l)$ since 
	\begin{align}
		g_1(l) &= \log \frac{W_1(l | 1)}{W_1(l | 0)}  
		\stackrel{(a)}{=} -  \log \frac{f_{L|B}(-l|1)}{f_{L|B}(-l | 0)} 
		\stackrel{(b)}{=} - g_0(-l) \EqCo
	\end{align}
	where $(a)$ follows from $W_1(l | b)
	= f_{L|B}(-l|\bar{b})$ and $(b)$ follows from $W_0(l | b) =
	f_{L|B}(l|b)$. 
	For exact L-values, we always have $g_0(l) = l$
	\cite[Th.~3.10]{Alvarado15_Book} which
	implies that $\MI_{\Lexmix} = \MI_{\Lex}$.
\end{proof}

\begin{remark}
As mentioned above, the result that $\MI_{\Lexmix} = \MI_{\Lex}$ was
already proved in~\cite[Th.~2]{Hou2003}. It also follows directly
from~\theref{theo:gmi_symmetrization} using the equivalence of the
\gls{GMI} and the \gls{MI} for exact L-values
in~\eqref{eq:mi_ex}.\footnote{As pointed out in~\cite[Th.~2]{Hou2003},
after symmetrization we lose an opportunity to optimize the input
distribution, and therefore, the symmetrization may decrease the
channel capacity. Channel capacity, however, is not studied in this
paper.}
\end{remark}

Observe that for exact L-values, we always have $\MI_{\Lexmix} =
\MI_{\Lex}$ regardless of whether the channel
$f_{\Lex|B}(\cdot|\cdot)$ is output-symmetric or not. Obviously, for
max-log L-values with an output-symmetric conditional PDF
$f_{\Lml|B}(\cdot|\cdot)$ we have $\MI_{\Lmlmix} = \MI_{\Lml}$, since
$g(l) = - g(-l)$. The numerical results presented
in~\secref{sec:num_res} suggest that if $f_{\Lml|B}(\cdot|\cdot)$ is
not output-symmetric, the correction function does not satisfy $g(l) =
- g(-l)$, and, hence, the inequality is strict. However, a proof for
this observation does not seem to be straightforward. 

The rate loss for symmetrized max-log L-values can be interpreted in
the following way. In order to achieve $\MI_{\Lml}$, the correction
function $g(\cdot)$ needs to be applied to the L-values prior to
decoding~\eqref{eq:decoder_general}. Hence, the information about the
asymmetry of the L-values is exploited by means of $g(\cdot)$. This
information is lost after the symmetrization unless $g(l) = -g(-l)$,
which causes the decrease of the mutual information.

From this analysis and \theref{theor:lossless_patterns}, we conclude
that the losses observed in~\cite[Fig.~4]{Alvarado07d}
and~\cite[Fig.~2]{Szczecinski07f} come partly from the L-value
symmetrization for constellations larger than $16$-ary \gls{QAM},
whereas the loss for $16$-\gls{QAM} is solely due to the
symmetrization and not due to the max-log approximation. This is in
contrast to the discussion included in
\cite{Alvarado07d,Szczecinski07f}, where the loss is attributed solely
to the max-log approximation. 

\subsection{Channel Mixing}

Channel mixing is another popular operation that is often assumed in
order to simplify the analysis of BICM systems~\cite{Xie12j}. Channel
mixing can be visualized in~\figref{fig:mixing}, where in addition to
the $m$-input $m$-output channel, an interleaver $\pi$ and a
deinterleaver $\pi^{-1}$ are introduced. The interleaver randomly and
uniformly assigns the input bits to the channel inputs and the
deinterleaver performs the reverse operation. Similarly to the
previous section, the random assignments of the bits are known to the
transmitter and the receiver. However, they are unknown to the
encoder and decoder.
For mixed channels, the $m$ L-values $\Lmix_1$, \dots, $\Lmix_m$ from
different bit positions have the same distribution and hence can be
treated equally, where the \gls{PDF} for all $j \in [m]$ is given
by
\begin{align}\label{eq:pdf_mixed}
	f_{\Lmix_j|\tilde{B}_j}(l|b) = \frac{1}{m}\sum_{j=1}^{m} f_{L_j|B_j}(l|b).
\end{align}

\begin{figure}[t]
	\centering
	\begin{center}
		\includegraphics{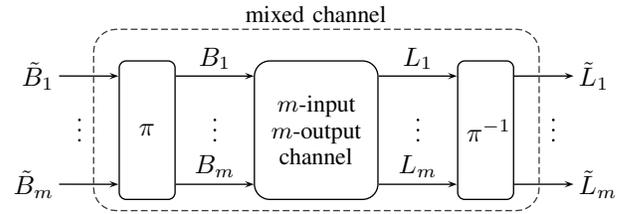}
	\end{center}
	\caption{Block diagram showing the mixing of bit positions. The interleaver $\pi$
	reorders the bit levels.}
	\label{fig:mixing}
\end{figure}

It is often said that channel mixing does not increase the BICM MI,
cf.~\cite[Th.~2]{Xie12j}, which is obvious from the data
processing inequality. In the following, we show that channel mixing
does not reduce the BICM GMI either. As in the case of the bit-level
symmetrization, the effect on the BICM MI depends on whether
exact or max-log L-values are used.

\begin{theorem}\label{theor:mixed}
For any one-dimensional constellation and any labeling, channel mixing does not affect the BICM GMI. In fact, we have
	\begin{align}
		\label{eq:gmis_mixed}
		\HBICMGMI_{\bLmix} = \BICMGMI_{\bLmix} = \BICMGMI_{\bL}.
	\end{align}
\end{theorem}
\begin{proof}
Using the definition of the harmonized GMI in~\eqref{eq:bicm_s_gmi}
and the fact that the mixed L-values all have the same distribution we
can write
	\begin{align}
	\HBICMGMI_{\bLmix} &= m - \sum_{j = 1}^{m} \inf_{s_j > 0} \Expect{\log 		\left( 1 + e^{-s_j\check{\tilde{B}}_j \Lmix_j} \right)}\nonumber\\
	&=m - \inf_{s_1 > 0} \sum_{j = 1}^{m}\Expect{\log 		\left( 1 + e^{-s_1\check{\tilde{B}}_j \Lmix_j} \right)}\label{eq:mix_step}
	\end{align}
	\begin{align}
	&=m -\inf_{s_1 > 0} \frac{m}{2} \sum_{b \in \mathcal{B}}\int_{-\infty}^{\infty}{ f_{\Lmix_1|\tilde{B}_1}(l|b) \log\left( 1 + e^{-s_1\check{b} l} \right) \,\,\mathrm{d}l}\nonumber\\
	&\stackrel{\eqref{eq:pdf_mixed}}{=}m	-\inf_{s_1 > 0} \frac{1}{2} \sum_{b \in \mathcal{B}}\int_{-\infty}^{\infty}{ \sum_{j=1}^{m} f_{L_j|B_j}(l|b) \log\left( 1 + e^{-s_1\check{b} l} \right) \,\,\mathrm{d}l}\nonumber\\
	 &=m - \inf_{s_1 > 0} \sum_{j=1}^{m}\frac{1}{2} \sum_{b \in \mathcal{B}}{\int_{-\infty}^{\infty} f_{L_j|B_j}(l|b) \log\left( 1 + e^{-s_1\check{b} l} \right) \,\,\mathrm{d}l}\nonumber\\
	 & = m - \inf_{s_1 > 0}\sum_{j = 1}^{m}  \Expect{\log 		\left( 1 +
	 e^{-s_1\check{B}_j L_j} \right)} \label{eq:last_step}\EqCo
	\end{align}
	where~\eqref{eq:mix_step} proves the first equality in~\eqref{eq:gmis_mixed} and~\eqref{eq:last_step} proves the second equality.
\end{proof}
Although channel mixing does not affect the BICM GMI, it reduces the
harmonized GMI for max-log L-values, i.e., $\HBICMGMI_{\bLmlmix} \leq
\HBICMGMI_{\bLml}$ with equality if and only if all bit-level
\glspl{GMI}~\eqref{eq:gmi_general} are minimized by the same value of
$s$. The intuitive explanation is that the bit-level information can
no longer be used by the encoder to construct $\setC$ as a product
code (see~Appendix~\ref{Appendix.theor:harmonized_gmi}).

Similarly to Corollary \ref{corol:symmetr}, the effect of channel
mixing on the mutual information is given as follows. 

\begin{corollary}\label{corol:mix}
	For any one-dimensional constellation and any labeling, we have
	$\BICMMI_{\bLexmix} = \BICMMI_{\bLex}$. Furthermore, we have $\BICMMI_{\bLmlmix} \leq
	\BICMMI_{\bLml}$ with equality if and only if the correction
	function in \eqref{eq:optimal_correction} is the same for all bit positions, i.e., $g_{j'}(l) = g_j(l)$ for all $j$, $j' \in [m]$. 
\end{corollary}
\begin{proof} 
	The proof is similar to the proof of Corollary \ref{corol:symmetr}
	and follows from Theorem
	\ref{theorem:mutual_information_convexity} in Appendix \ref{Appendix.theor:mutual_information_convexity}.
\end{proof}
The numerical results in the next section suggest that the
correction functions are different for non-equivalent patterns and,
hence, the inequality is strict for all one-dimensional constellations
with any labeling.  However, similarly to the channel symmetrization,
a proof for this observation does not seem to be straightforward. 

\section{Numerical Examples}\label{sec:num_res}
\newcommand{\loss}{\mathsf{L}}

So far, the information loss was characterized by inequalities. In
this section, however, we want to compare different information rates
quantitatively. To that end, we observe that all considered
information rates are strictly increasing functions of the \gls{SNR}. Hence, if
$R = \phi(\SNR)$, where $R$ is a generic information rate, then there
exists $\phi^{-1}(R) = \SNR$. Consider two rates $R_1$ and $R_2$ with
the corresponding functions $\phi_1$ and $\phi_2$ and assume that the
\gls{SNR} is expressed in dB. The loss of $R_2$ with respect to $R_1$
is defined as $\loss(R_1) = \phi^{-1}_2(R_1) - \phi^{-1}_1(R_1)$. The
loss $\loss$ can be graphically interpreted as the horizontal distance
between the curves $\phi_1(\SNR)$ and $\phi_2(\SNR)$ plotted over the
\gls{SNR} in dB for a particular value of the rate. We remark that
only rates that have the same range can be compared in terms of
$\loss$.

We first present numerical examples for the bit-level analysis to
illustrate the inequalities in~\eqref{ineq.pattern}. We consider the
equally spaced $4$-PAM constellation with the three nonequivalent
patterns defined in~\secref{sec:lossless_patterns}.
\figref{fig:bit_level_gmi} shows the loss in dB for different
achievable rates in~\eqref{ineq.pattern} with respect to (w.r.t.)
$\MI_{\Lex}$ as a function of the information rate in \gls{bpcu}. The
solid lines show the loss for the $\GMI_{\Lml}$. It can be seen that
the GMI for max-log L-values is always inferior to the MI for exact
L-values. An interesting behavior of the GMI for the pattern
$\boldsymbol{p}_{\mathrm{II}}$ at asymptotically low SNR is that the
loss grows unboundedly when the rate goes to zero (or equivalently,
when the SNR tends to zero).

The dashed line in~\figref{fig:bit_level_gmi} shows the loss for
$\MI_{\Lml}$. There is only one dashed line in the figure as two of
the three patterns are information lossless. Finally, the dash-dotted
line illustrates the effect of the symmetrization on the MI for the
pattern $\boldsymbol{p}_{\mathrm{II}}$ and confirms
Corollary~\ref{corol:symmetr}. We also observe that the
inequality in Corollary~\ref{corol:symmetr} in this case is strict.
We remark that, for $4$-PAM, $\boldsymbol{p}_{\mathrm{II}}$ is the
only pattern that gives such L-values (hence, only one dash-dotted
curve is shown in~\figref{fig:bit_level_gmi}).

\begin{figure}
	\centering \includegraphics[]{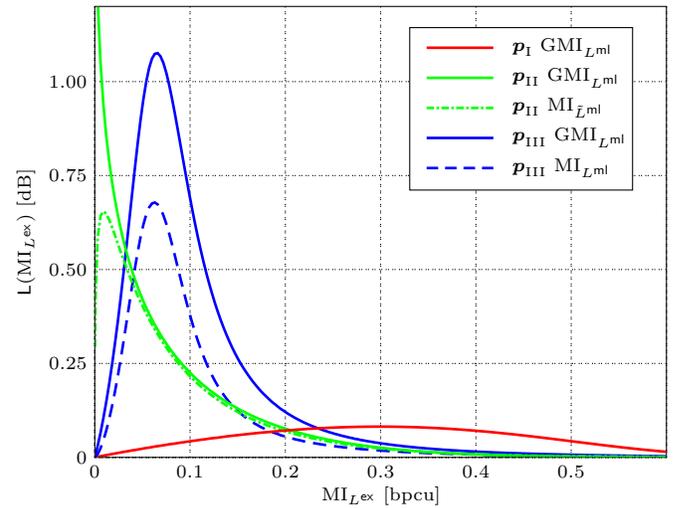}
	\caption{Different losses for the three patterns for an equally
	spaced $4$-PAM constellation. $\boldsymbol{p}_{\mathrm{I}}
	= [0, 0,
	1, 1]$, $\boldsymbol{p}_{\mathrm{II}} = [0, 1, 1, 0]$, and
	$\boldsymbol{p}_{{\mathrm{III}}} = [0, 1, 0, 1]$.}
	\label{fig:bit_level_gmi}
\end{figure}

As for the BICM analysis, we first consider an equally spaced $4$-PAM
constellation labeled with the NBC.  \figref{fig:BICM_inequalities}
shows the loss for different achievable rates w.r.t.~to the
$\BICMMI_{\bLex}$ as a function of the information rate. We note that
the third inequality in~\eqref{ineq.labeling.maxlog} becomes an
equality, i.e., $\HBICMGMI_{\bLml} = \BICMGMI_{\bLml}$, for $R \approx
0.137$.  The solid lines illustrate the inequalities
in~\eqref{ineq.labeling.maxlog}. The red, green, and 
blue curves are in the order of increasing loss. The dashed red line
shows the effect of channel mixing and is above the solid red line in
agreement with Corollary~\ref{corol:mix}. Moreover, this suggests that
the inequality in~Corollary~\ref{corol:mix} is strict.

In \figref{fig:BICM_8pam}, we show the loss associated with the use
of the max-log approximation for the $8$-PAM constellation labeled
with the BRGC (see Fig.~\ref{fig:4-pam_brgc}\subref{fig:subplot_8pam}). The solid and dashed lines correspond to the GMI and
the MI, respectively. Red, green, and blue lines show the loss for
the three patterns in the BRGC, whereas the magenta lines show the
loss for the entire labeling. Similarly to the GMI for $\boldsymbol{p}_{\mathrm{II}}$ in \figref{fig:bit_level_gmi}, the loss of the GMI for the patterns $\boldsymbol{p}_{2}$ and $\boldsymbol{p}_{3}$ goes to infinity as the rate goes to zero.

\begin{figure}
	\centering
	\includegraphics[]{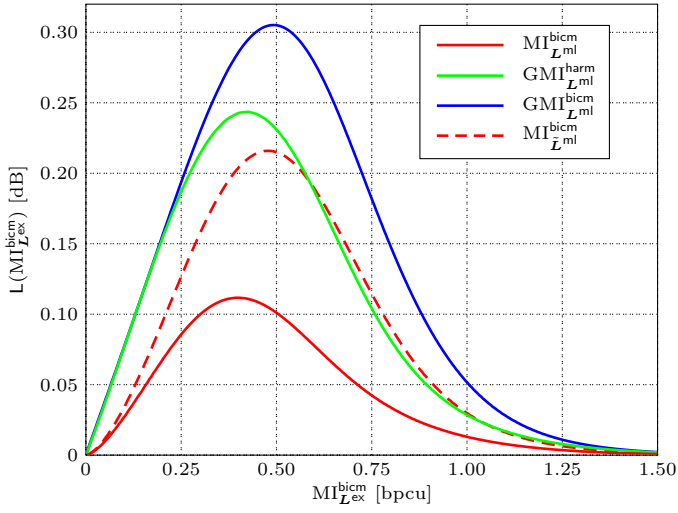}
	\caption{Different achievable rates for $4$-PAM with the NBC.}
	\label{fig:BICM_inequalities}
\end{figure}

\begin{figure}
	\centering
	\includegraphics[]{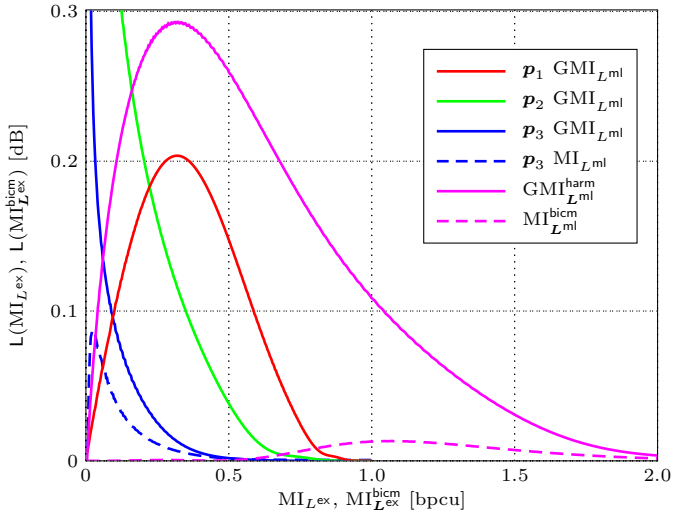}
	\caption{Different achievable rates for $8$-PAM with the BRGC.
	$\quad\boldsymbol{p}_1 = [0, 0, 0, 0, 1, 1, 1, 1]$, $\boldsymbol{p}_2 = [0, 0, 1, 1, 1, 1, 0, 0]$, and $\boldsymbol{p}_3 = [0, 1, 1, 0, 0, 1, 1, 0]$.}
	\label{fig:BICM_8pam}
\end{figure}

It can be seen from the figures that the losses can be quite large
for low information rates (i.e., low SNR). Furthermore, the losses can
even go to infinity as the rate goes to zero. However, for moderately
high rates, the losses are small. For instance, for 8-PAM with any
labeling and assuming a rate of $R = m-1$ (as advised
in~\cite{Unger82jan} to be used in coded modulation), the loss of
$\BICMMI_{\bLml}$ and $\HBICMGMI_{\bLml}$ w.r.t. $\BICMMI_{\bLex}$  does
not exceed $0.013$~dB and $0.024$~dB, respectively. Therefore, we
conclude that the differences between these achievable rates are
negligible from a practical viewpoint.

\section{Conclusions}
\label{sec:conclusion}

In this paper, we studied achievable rates of a BICM decoder
with both exact and max-log L-values for $M$-PAM
constellations. We showed that the max-log
approximation is not information lossy in some cases.
Furthermore, when exact L-values are considered, seemingly different
quantities, e.g., BICM MI, BICM GMI, or harmonized GMI are shown to give
the same achievable rate. This is not the case for max-log
L-values as these quantities are different. For high SNR, however, the differences
between them become negligible, which justifies the use
of the max-log approximation in practical systems and the considered
processing techniques for their analysis.

\appendices

\section{Proof of Lemma \ref{lemma:exact_symmetric}}
\label{Appendix.theor:exact_symmetric}

\newcommand{\R}{{\mathbb{R}}}
\newcommand{\W}{{\mathcal{W}}}
\newcommand{\mya}{d}
\newcommand{\myb}{t}

\newcommand{\myM}{{\tilde{M}}}

A symmetric constellation	and pattern are easily shown to be
sufficient for the exact L-value to be symmetric by using the
definition of $\lex(y)$ in \eqref{eq:Lvalue_exact}. To show that they
are also necessary we argue as follows.  Assume for a moment that the
symmetry point is the origin, i.e., $y_0 = 0$. We also define $v = e^{\rho}$, where $v>1$ since $\rho>0$. We denote the subconstellation of points
labeled with a $1$ by $\mya_1 < \dots < \mya_{\myM}$ and the
subconstellation of points labeled with a $0$ by $\myb_1 < \dots <
\myb_{\myM}$, where $\myM = M/2$. Then, using \eqref{eq:Lvalue_exact} and
$\lex(y) = \lex(-y)$, we find that, for all $y\in\R$,
\begin{equation}
	\label{eq:identity1}
\sum_i v^{2y\mya_i-\mya_i^2} \sum_j v^{-2y\myb_j-\myb_j^2}
= \sum_i v^{-2y\mya_i-\mya_i^2} \sum_j v^{2y\myb_j-\myb_j^2}.
\end{equation}
Substituting $\myM+1-i$ for $i$ in the first sum and $\myM+1-j$ for $j$ in the last sum yields
\begin{multline} \label{eq:identity2}
\sum_i v^{2y\mya_{\myM+1-i}-\mya_{\myM+1-i}^2} \sum_j v^{-2y\myb_j-\myb_j^2} \\
= \sum_i v^{-2y\mya_i-\mya_i^2} \sum_j v^{2y\myb_{\myM+1-j}-\myb_{\myM+1-j}^2}
\end{multline}
or, rearranging terms,
\begin{multline} \label{eq:identity3}
\sum_{i,j} v^{2y(\mya_{\myM+1-i}-\myb_j)-\mya_{\myM+1-i}^2-\myb_j^2} \\
= \sum_{i,j} v^{2y(\myb_{ {\myM}+1-j}-\mya_i)-\myb_{{\myM}+1-j}^2-\mya_i^2}.
\end{multline}

As $y \to \infty$, the largest exponents on both sides of~\eqref{eq:identity3} have to be the same for the equality to hold. The largest exponents correspond to $i=j=1$ and thus
\begin{equation} \label{eq:leading-term}
	v^{2y(\mya_{\myM}-\myb_1)-\mya_{\myM}^2-\myb_1^2} = v^{2y(\myb_{\myM}-\mya_1)-\myb_{\myM}^2-\mya_1^2}, \quad \forall y \in \R \EqCo
\end{equation}
which yields
\begin{align}
	\mya_{\myM}-\myb_1 &= \myb_{\myM}-\mya_1, \\
	\mya_{\myM}^2+\myb_1^2 &= \myb_{\myM}^2+\mya_1^2
\end{align}
or equivalently
\begin{align}
	\mya_{\myM}-\myb_{\myM} &= \myb_1-\mya_1, \label{eq:abdiff1} \\
	\mya_{\myM}^2-\myb_{\myM}^2 &= -\myb_1^2+\mya_1^2. \label{eq:a2b2diff1}
\end{align}
Factorizing \eqref{eq:a2b2diff1},
\begin{equation}
	(\mya_{\myM}+\myb_{\myM})(\mya_{\myM}-\myb_{\myM}) = -(\myb_1+\mya_1)(\myb_1-\mya_1).
\end{equation}
Dividing both sides by $\myb_1-\mya_1$, which by assumption of distinct constellation points in~\secref{sec:modulator} is nonzero, and using \eqref{eq:abdiff1} yields
\begin{equation} \label{eq:absum1}
	\mya_{\myM}+\myb_{\myM} = -\myb_1-\mya_1.
\end{equation}
Combining \eqref{eq:abdiff1} with \eqref{eq:absum1} yields
\begin{equation}\label{eq:absym}
	\mya_1 = -\mya_{\myM} \textrm{ and } \myb_1 = -\myb_{\myM}.
\end{equation}

We will now prove by contradiction that $\mya_i = -\mya_{{\myM}+1-i}, \forall i$. To this end, suppose the opposite, i.e., that there exists an integer $k\ge 2$ such that
\begin{equation}\label{eq:asym}
	\mya_i = -\mya_{{\myM}+1-i}, i=1,\ldots,k-1 \textrm{ and } \mya_k \ne
	-\mya_{{\myM}+1-k}.
\end{equation}
Regarding the relation between  $\myb_j$ and $-\myb_{{\myM}+1-j}$, we have already proven in \eqref{eq:absym} that they are the same for $j=1$. For $j\ge2$, we will distinguish between two cases, one of which must be true. However, both will be shown to lead to contradictions in combination with \eqref{eq:asym}, which can only mean that \eqref{eq:asym} is false.

\emph{Case 1:}
Suppose that
\begin{equation} \label{eq:bsym1}
	\myb_j = -\myb_{{\myM}+1-j}, \forall j.
\end{equation}
Then \eqref{eq:identity2} and \eqref{eq:bsym1} together yield
\begin{equation}
	\sum_{i=1}^{\myM} v^{2y\mya_{{\myM}+1-i}-\mya_{{\myM}+1-i}^2} =
	\sum_{i=1}^{\myM} v^{-2y\mya_i-\mya_i^2}, \quad \forall y \in \R.
\end{equation}
Cancelling terms using \eqref{eq:asym},
\begin{equation} \label{eq:identity4}
	\sum_{i=k}^{{\myM}+1-k} v^{2y\mya_{{\myM}+1-i}-\mya_{{\myM}+1-i}^2} =
	\sum_{i=k}^{{\myM}+1-k} v^{-2y\mya_i-\mya_i^2}, \quad \forall y \in \R.
\end{equation}
Again considering $y\to \infty$, the largest exponents occur for $i=k$ on
both sides. However, since $\mya_k \ne \mya_{{\myM}+1-k}$ by \eqref{eq:asym}, these exponents are unequal, which contradicts \eqref{eq:identity4}. It can be concluded that Case 1, defined by the assumption \eqref{eq:bsym1}, cannot be true.

\emph{Case 2:}
Suppose that there exists an integer $\ell \ge 2$ such that
\begin{equation}\label{eq:bsym2}
	\myb_j = -\myb_{{\myM}+1-j}, j=1,\ldots,\ell-1 \textrm{ and } \myb_\ell \ne
	-\myb_{{\myM}+1-\ell}.
\end{equation}
By \eqref{eq:asym} and \eqref{eq:bsym2}, the terms in
\eqref{eq:identity3} for which $i<k, \,j<\ell$ and $i>\myM +1 - k, \,j> \myM +1 - \ell$ cancel each
other. Defining $\W \triangleq \{(i,j): k\le i \le \myM + 1 - k \text{ or } \ell \le j \le \myM + 1 - \ell \}$, \eqref{eq:identity3} simplifies into
\begin{multline} \label{eq:identity5}
	\sum_{(i,j) \in\W} v^{2y(\mya_{{\myM}+1-i}-\myb_j)-\mya_{{\myM}+1-i}^2-\myb_j^2} \\
	= \sum_{(i,j) \in\W} v^{2y(\myb_{{\myM}+1-j}-\mya_i)-\myb_{{\myM}+1-j}^2-\mya_i^2}, \quad \forall y \in \R.
\end{multline}
When $y\to \infty$, the largest exponents in both sums in~\eqref{eq:identity5} correspond to the pairs $(i,j) \in \W$ for which
$\mya_{{\myM}+1-i}-\myb_j$ and $\myb_{{\myM}+1-j}-\mya_i$, respectively, are maximum. Since
$\mya_1 \le \mya_i \le \mya_{\myM}$ and $\myb_1 \le \myb_j \le \myb_{\myM}$ for all $i,j$, 
\begin{align}
	\max{(i,j)\in\W} \mya_{{\myM}+1-i}-\myb_j &= \mathrm{max}\{
	\mya_{\myM}-\myb_\ell, \mya_{{\myM}+1-k}-\myb_1 \}, \label{eq:maxab1} \\
	\max{(i,j)\in\W} \myb_{{\myM}+1-j}-\mya_i &= \mathrm{max}\{ \myb_{\myM}-\mya_k,
	\myb_{{\myM}+1-\ell}-\mya_1 \}. \label{eq:maxab2}
\end{align}
These maxima must be equal for \eqref{eq:identity5} to hold for large $y$.

By \eqref{eq:absym}, \eqref{eq:asym} and \eqref{eq:bsym2}, $\mya_{\myM}-\myb_\ell \ne
\myb_{{\myM}+1-\ell}-\mya_1$ and $\myb_{\myM}-\mya_k \ne \mya_{{\myM}+1-k}-\myb_1$. This leaves only two possibilities to equate the right-hand sides of \eqref{eq:maxab1} and \eqref{eq:maxab2}: Either
\begin{equation} \label{eq:abba1}
	\mya_{{\myM}+1-k}-\myb_1 < \mya_{\myM}-\myb_\ell = \myb_{\myM}-\mya_k > \myb_{{\myM}+1-\ell}-\mya_1
\end{equation}
or
\begin{equation} \label{eq:abba2}
	\mya_{\myM}-\myb_\ell < \mya_{{\myM}+1-k}-\myb_1 = \myb_{{\myM}+1-\ell}-\mya_1 >
	\myb_{\myM}-\mya_k.
\end{equation}
Equating the dominating terms of \eqref{eq:identity5} if \eqref{eq:abba1} is true yields
\begin{equation}
	v^{2y(\mya_{\myM}-\myb_\ell)-\mya_{\myM}^2-\myb_\ell^2} =
	v^{2y(\myb_{\myM}-\mya_k)-\myb_{\myM}^2-\mya_k^2}.
\end{equation}
In analogy with \eqref{eq:leading-term}--\eqref{eq:absym}, this
equality implies $\mya_k = -\mya_{\myM}$, which contradicts $\mya_k > \mya_1 =
-\mya_{\myM}$. Analogously, \eqref{eq:abba2} implies $\mya_{{\myM}+1-k} =
-\mya_1$, which contradicts $\mya_{{\myM}+1-k} < \mya_{\myM} = -\mya_1$. Hence, neither \eqref{eq:abba1} nor \eqref{eq:abba2} can be true. It can be concluded that Case 2, defined by the assumption \eqref{eq:bsym2}, cannot be true.

Since both Case 1 and Case 2 lead to contradictions, the assumptions
\eqref{eq:bsym1} and \eqref{eq:bsym2} are both false. This proves that
\eqref{eq:asym} is false, which implies $\mya_i = -\mya_{{\myM}+1-i},
\forall i$. Finally, $\myb_j = -\myb_{{\myM}+1-j}, \forall j$ follows
because $\mya_i$ and $\myb_j$ have equivalent roles in \eqref{eq:identity1}.
The case $y_0 \neq 0$ follows by
applying a shift in the coordinate system. 

\section{Proof of \theref{theor:lossless_patterns}}\label{Appendix.theor:lossy_patterns}

We first recall several facts about the max-log L-value $\lml(y)$ in
\eqref{eq:Lvalue_maxlog} and the exact L-value $\lex(y)$ in
\eqref{eq:Lvalue_exact} which will be used later on in the proof.

\begin{enumerate}
	\item[(F1)] The max-log L-value is a continuous piecewise linear
		function of the observation, where the slope of the linear pieces
		changes at the midpoints between neighboring constellation points
		labeled with the same bit \cite{Alvarado07d}.

	\item[(F2)] The max-log L-value has zero-crossings at midpoints
		between adjacent constellation points labeled with different bits
		\cite{Alvarado07d}. 

	\item[(F3)] The exact L-value is an analytic function. Indeed, the
		nominator and the denominator in~\eqref{eq:Lvalue_exact} are sums
		of exponential functions and therefore, they are analytic
		functions, as is their ratio. The logarithm of an analytic
		function is also analytic. 

\end{enumerate}

Note that (F2) implies that for any pattern, the max-log L-value has
at least one zero-crossing. 

The proof is structured as follows. We start by showing that for the
only two patterns that induce one zero-crossing (i.e.,
$\boldsymbol{p}_{\mathrm{I}} = [\mathbf{0}_{M/2}, \mathbf{1}_{M/2}]$
and $\mathrm{inv}(\boldsymbol{p}_{\mathrm{I}})$), the max-log L-value
is information lossless, regardless of the constellation. We then
proceed by considering patterns that induce \emph{exactly} two
zero-crossings. Such patterns are of the form
\begin{align}
	\label{eq:pattern_two_crossings}
	[\;\,\underbrace{b\,,\, \dots\,,\, b}_{\text{$a$}},\;
	\underbrace{\,\bar{b}\,,\, \dots\,,\,
	\bar{b}}_{\text{$M/2$}}, \underbrace{\,b\,,\, \dots\,,\, b\,}_{\text{$M/2-a$}}]
\end{align}
where $1 \leq a < M/2$. We show that a necessary and sufficient
condition for the max-log L-value to be information lossless is that
the pattern is equivalent to $\boldsymbol{p}_{\mathrm{II}} =
[\mathbf{0}_{M/4}, \mathbf{1}_{M/2}, \mathbf{0}_{M/4}]$ and the
constellation is symmetric. Lastly, we consider patterns that induce
more than two zero-crossings and show that in this case, the max-log
L-value can never be information lossless.

\subsection{One Zero-Crossing}

Consider the max-log L-value $\lml(y)$ in~\eqref{eq:Lvalue_maxlog} for
an arbitrary constellation and the pattern
$\boldsymbol{p}_{\mathrm{I}} = [\mathbf{0}_{M/2}, \mathbf{1}_{M/2}]$.
For a certain value $y$, let $a_t = \argmin{a \in \setS_{0}}{(y-a)^2}$
and $a_s = \argmin{a \in \setS_{1}}{(y-a)^2}$, where $\setS_0$ and
$\setS_1$ are the subconstellations with points labeled with 0 and 1,
respectively. The max-log L-value can then be written as $\lml(y) =
2\rho(a_s - a_t)y + \rho(a_t^2 - a_s^2)$, where $a_s$ and $a_t$ are piece-wise constant functions of $y$. Due to the structure of the
pattern, $a_s > a_t$ for any value of $y$, which implies that the
derivative $\mathrm{d}\lml(y)/\mathrm{d}y$ is positive whenever it
exists. (It does not exist whenever the slope of the linear pieces
changes, see (F1) above.) This, together with the fact that the
max-log L-value is a continuous function of the observation,
guarantees that $\lml(y)$ is strictly increasing. Therefore,
$\lml(y)$ is invertible, i.e., the observation $y$ can be recovered
from $\lml(y)$.  Hence, since the exact L-value is a function of $y$, it can be obtained from the
max-log L-value. The same is true for the pattern
$\mathrm{refl}(\boldsymbol{p}_{\mathrm{I}})$, in which case $\lml(y)$
is strictly decreasing. This claim holds for any constellation,
not necessarily symmetric ones.

\subsection{Two Zero-Crossings}

Next, consider an arbitrary constellation with a pattern that induces
\emph{exactly} two zero-crossings, i.e., a pattern of the form in
\eqref{eq:pattern_two_crossings}. Let $y_q$ and $y_r$ denote these two
zero-crossings situated between the constellation points $a_q$ and
$a_{q+1}$, and $a_r$ and $a_{r+1}$, respectively. This implies that
$p_q = p_{r+1} = \bar{p}_{q+1} = \bar{p}_{r}$. Without loss of
generality, we assume $p_q = 0$. This is illustrated
in~\figref{fig:invert_proof}.
\begin{figure}
	\centering \includegraphics{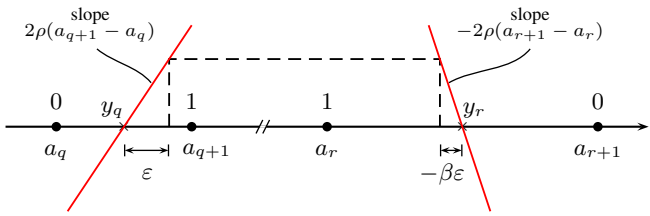} \caption{Schematic representation
	of two neighboring zero-crossings for the max-log L-value $\lml(y)$
	(red lines).}
	\label{fig:invert_proof}
\end{figure}
From the figure, we see that there exists an $\eps > 0$, such that
\begin{align}
	\lml(y_q + \gamma) = \lml(y_r - \beta \gamma) \quad \text{ for }
	\gamma \in [-\eps, +\eps] \EqCo
\end{align}
where 
\begin{equation}\label{eq:beta}
\beta = \frac{a_{q+1} - a_{q}}{a_{r+1} - a_{r}} > 0.
\end{equation}

According to Lemma~\ref{lemma:min_suf_stat}, for the max-log
L-value to be information lossless, the exact L-value should be
recoverable from the max-log L-value, i.e., the exact L-value has to
satisfy the condition
\begin{align}\label{eq:sym_lex_part}
	\lex(y_q + \gamma) = \lex(y_r - \beta \gamma) \quad \text{ for }
	\gamma \in [-\eps, +\eps].
\end{align}
If this condition is not satisfied, more than one value of $\lex(y)$ will
correspond to one value of $\lml(y)$. The
condition~\eqref{eq:sym_lex_part} can be rewritten as 
\begin{equation}
	v(\gamma) = 0 \quad \text{ for } \gamma \in [-\eps, + \eps]\EqCo
\end{equation}
where $v(\gamma) = \lex(y_q + \gamma) - \lex(y_r - \beta \gamma)$ is
an analytic function (see (F3)). If an
analytic function is zero on an interval, it has to be zero everywhere
it is defined, i.e., $v(\gamma) = 0$ for $\gamma \in \realR$ or
\begin{align}\label{eq:sym_lex_part1}
	\lex(y_q + \gamma) = \lex(y_r - \beta \gamma) \quad \text{ for }
	\gamma \in \mathbb{R}.
\end{align}
Using the substitution $y = \gamma - (y_r - y_q)/(1+\beta)$, we
can rewrite \eqref{eq:sym_lex_part1} as
\begin{align}
	\label{eq:sym_lex_part2}
	\lex(y_0 + y) = \lex(y_0 - \beta y) \quad \text{ for }
	y \in \mathbb{R}.
\end{align}
where $y_0 = (y_r + \beta y_q)/(1+\beta)$. We next argue that we must
have $\beta = 1$. Indeed, since $\lex(y)$ is an analytic function,
it follows from~\eqref{eq:sym_lex_part2} that 
\begin{equation}
	\lex{}^{(k)}(y_0 + y) = (-\beta)^{k} \lex{}^{(k)}(y_0 - \beta y).
\end{equation}
for $k \in \mathbb{N}$, where $\lex{}^{(k)}$ denotes the $k$th
derivative with respect to $y$. In particular, for $y = 0$, we get 
\begin{equation}
	\lex{}^{(k)}(y_0) = (-\beta)^{k} \lex{}^{(k)}(y_0) \quad \text{  for } k \in \mathbb{N} \EqCo
\end{equation}
which can only hold if either $\beta = 1$ or $\lex{}^{(k)}(y_0) = 0$
for all $k \in \mathbb{N}$. Assume that the latter holds. We can write
$\lex(y_0 + y)$ as a Taylor expansion around $y_0$ as
\begin{equation}
\lex(y_0 + y) = \sum_{k=0}^{\infty} \frac{\lex{}^{(k)}(y_0)}{k!} y^n =
\lex(y_0).
\end{equation}
However, since $\lex(y)$ cannot be constant for all $y \in \mathbb{R}$,
we therefore conclude that we must have $\beta = 1$, or
\begin{align}\label{eq:sym_lex_part4}
	\lex(y_0 + y) = \lex(y_0 - y) \quad \text{ for }
	y \in \mathbb{R}.
\end{align}
This means that the exact L-value has to be a symmetric function
around $y_0 = (y_r + y_q)/2$, i.e., the midpoint between the two
zero-crossings. According to Lemma \ref{lemma:exact_symmetric}, the
constellation therefore has to be symmetric around $y_0$ and the
pattern has to be $\boldsymbol{p} = [\mathbf{0}_{M/4},
\mathbf{1}_{M/2}, \mathbf{0}_{M/4}]$ in order to satisfy
$\boldsymbol{p} = \mathrm{refl}(\boldsymbol{p})$.  

Consider now a symmetric constellation around $y_0$ and the pattern
$\boldsymbol{p}_{\mathrm{II}} = [\mathbf{0}_{M/4}, \mathbf{1}_{M/2},
\mathbf{0}_{M/4}]$. In this case, both the exact and the max-log
L-value are symmetric functions around $y_0$, see Lemma
\ref{lemma:exact_symmetric} and Remark \ref{remark:maxlog_symmetric}.
To show that the max-log L-value is information lossless, it is
therefore enough to show that $|y-y_0|$ is recoverable from $\lml(y)$.
This can be done by showing that $\lml(y)$ is strictly decreasing 
for $y \ge y_0$, similarly as before.

\subsection{More Than Two Zero-Crossings}

Lastly, consider an arbitrary constellation $\mathcal{S}$ with a
pattern such that the max-log L-value $\lml(y)$ has \emph{more than
two} zero-crossings. Now, consider any two pairs of neighboring
zero-crossings and denote their (distinct) midpoints by $y_0$ and
$y_0'$, respectively. Under the assumption that the max-log L-value is
information lossless and using the same arguments as above, we find
that the exact L-value should satisfy both
\begin{align}
	\lex(y_0 + y) = \lex(y_0 - y) \quad \text{ for } y \in \mathbb{R} \EqCo
\end{align}
and
\begin{align}
	\lex(y_0' + y) = \lex(y_0' - y) \quad \text{ for } y \in \mathbb{R}.
\end{align}
In the light of Remark \ref{remark:exact_symmetric_unique}, we
conclude that it is not possible to satisfy both conditions
simultaneously, i.e., we conclude that a max-log L-value with more
than two-crossings cannot be information lossless.

\newif\ifshow
\showfalse

\ifshow

\section{Old}

\begin{lemma}
	If the max-log L-value has more than two zero-crossings, it cannot
	be information lossless, i.e., there cannot exist a function such
	that $\Lex = f(\Lml)$. 
\end{lemma}

\begin{proof}

Consider an arbitrary constellation $\mathcal{S}$ with an arbitrary
pattern $\boldsymbol{p}$ such that the max-log L-value $\lml(y)$ has
\emph{at least} two zero-crossings. under the assumption that the
max-log L-value is information lossless (in the sense of Lemma
\ref{lemma:min_suf_stat}). An outline of the proof is given as
follows. 

\begin{itemize}
	\item We examine two neighboring zero-crossings and show that the
		max-log L-value satisfies a local symmetry condition.

	\item Using Lemma \ref{lemma:min_suf_stat}, we argue that if the
		exact L-value is recoverable from the max-log L-value, this local
		symmetry condition needs to be satisfied also for the exact
		L-value

	\item Using the fact that the exact L-value is an analytic function,
		we argue that if the symmetry condition is locally satisfied, it
		also has to be satisfied \emph{globally} 

	\item We then argue that the global symmetry condition for the exact
		L-value corresponds to $\lex(y_0 + y) = \lex(y_0-y)$ for $y \in
		\mathbb{R}$.

	\item We then repeat the same steps for two other neighboring
		zero-crossings (we know that there is at least one more
		zero-crossing) to obtain $\lex(y+y_0') = \lex(-y+y_0')$ with $y_0'
		\neq y_0$ as an additional condition for information losslessness

	\item In the light of Remark \ref{remark:exact_symmetric_unique}, we
		conclude that it is not possible to satisfy both conditions
		simultaneously, i.e., we conclude that a max-log L-value with more
		than two-crossings cannot be information lossless. 

\end{itemize}

Next, consider a pattern with at least two zero crossings as shown
in~\figref{fig:invert_proof}.  Let $y_q$ and $y_r$ denote any two
neighboring zero-crossings situated between the constellation points
$a_q$ and $a_{q+1}$, and $a_r$ and $a_{r+1}$, respectively. This
implies that $p_q = p_{r+1}  = \bar{p}_{q+1} = \bar{p}_{r}$. Without
loss of generality, we assume $p_q = 0$. Let $y_0 = (y_q + y_r)/2$ and
without loss of generality we assume $y_0 \ge 0$. This is illustrated
in~\figref{fig:invert_proof}.

The max-log L-value in~\eqref{eq:Lvalue_maxlog} is a continuous
piecewise linear function, where the slope of the pieces changes at
the midpoints between neighboring constellation points labeled with the same
bit. Since all constellation points are distinct, the slope of the
max-log L-value cannot change at the zero-crossing, i.e., there exists
an $\eps > 0 $ such that
\begin{align}\label{eq:lml_piece1}
	\lml(y) = 2\rho(a_{q+1} - a_{q})(y -y_q) \\\text{ for } y \in [y_q -
	\eps, y_q + \eps]
\end{align}
and
\begin{align}\label{eq:lml_piece2}
	\lml(y) = 2\rho(a_{r+1} - a_{r})(y_r - y) \\\text{ for } y \in [y_r
	- \eps, y_r + \eps].
\end{align}
The linear functions in~\eqsref{eq:lml_piece1}{eq:lml_piece2} are
shown with dashed lines in~\figref{fig:invert_proof}. Basically,
$\eps$ indicates that there are nonzero intervals of $y$ for which
$\lml(y)$ is a linear function that crosses the zero-level. One can
then identify the values of $y$ that give the same value of $\lml(y)$.
Based on~\eqsref{eq:lml_piece1}{eq:lml_piece2}, we can show that 
\begin{align}\label{eq:lml_almost_sym}
	\lml(y + y_0) = \lml(-\beta y +(2-\beta)y_0 + (\beta-1)y_r) \\\text{
	for } y \in [y_r-y_0 - \eps, y_r-y_0 + \eps]\EqCo
\end{align}
where 
\begin{equation}\label{eq:beta}
\beta = \frac{a_{r+1} - a_{r}}{a_{q+1} - a_{q}}.
\end{equation}
In the following, we show that $\beta$ has to equal one.

According to Lemma~\ref{lemma:min_suf_stat}, for the max-log approximation to be information lossless, the exact L-value should be recoverable from the max-log L-value, i.e., the exact L-value should satisfy the condition
\begin{multline}\label{eq:sym_lex_part}
\lex(y + y_0) = \lex(-\beta y +(\beta-k)y_0 + (\beta-1)y_r) \\\text{
for } y \in [y_r-y_0 - \eps, y_r-y_0 + \eps].
\end{multline}
If this condition is not satisfied, several values of $\lex(y)$ will correspond to one value of $\lml(y)$.

The nominator and the denominator of the exact L-value
in~\eqref{eq:Lvalue_exact} are sums of exponential functions and
therefore, they are analytic functions, as is their ratio. The
logarithm of an analytic function is analytic and hence the exact
L-value is an analytic function. The condition~\eqref{eq:sym_lex_part}
can be rewritten as 
\begin{equation}
	v(y) = 0, \text{ for } y \in [y_r-y_0 - \eps, y_{r}-y_0 +
	\eps]\EqCo
\end{equation}
where $v(y) = \lex(y_0 - y) - \lex(-\beta y +(2-\beta)y_0 +
(\beta-1)y_r)$ is an analytic function. If an analytic function is
zero on an interval, it has to be zero everywhere it is defined, i.e.,
$v(y) = 0$ for $y \in \realR$ or
\begin{equation}\label{eq:sym_lex}
\lex(y + y_0) = \lex(-\beta y +(2-\beta)y_0 + (\beta-1)y_r) \text{ for } y \in \realR.
\end{equation}
Since $\lex(y) \rightarrow \lml(y)$ $\forall y$ when $\SNR \rightarrow \infty$, the same should hold for the max-log L-value, i.e.,
\begin{equation}\label{eq:ml_lex} denoted
\lml(y + y_0) = \lml(-\beta y +(2-\beta)y_0 + (\beta-1)y_r) \text{ for } y \in \realR
\end{equation}
for $\SNR \rightarrow \infty$. Since $\SNR$ appears as a scaling factor in~\eqref{eq:Lvalue_maxlog}, \eqref{eq:ml_lex} in fact holds for any SNR values. The max-log L-value will then achieve the maximum value on the interval $[y_{q}, y_r]$ when the arguments of the right-hand side and the left-hand side in~\eqref{eq:ml_lex} are equal, i.e., at $y' = y'' + y_0$, where $y'' + y_0 = -\beta y'' +(2-\beta)y_0 + (\beta-1)y_r$. We can show that the maximum is achieved at
\[
y' = y_0 + \frac{\beta-1}{\beta+1}\frac{y_r - y_q}{2}.
\] 
 On the other hand, the maximum of the L-value should be at the midpoint between the constellation points $a_q$ and $a_{r+1}$ labeled with the zero bit since this is the point where the slope of the L-value in~\eqref{eq:Lvalue_maxlog} changes its sign, i.e., $y' = (a_{r+1}+a_{q})/2 = y_0 + \left( (a_{r+1}-a_{r}) - (a_{q+1} - a_{q})\right)/4$.
Therefore, the following should hold
\begin{equation}\label{eq:people_equal_one}
	\frac{\beta-1}{\beta+1}\frac{y_r - y_q}{2} = y_0 + \frac{ (a_{r+1}-a_{r}) - (a_{q+1} - a_{q})}{4}.
\end{equation}
Using the definition of $\beta$ in~\eqref{eq:beta}, one can easily
show that~\eqref{eq:people_equal_one} holds only for $\beta = 1$. We
thus conclude that 
\begin{align}
	\lml(y+y_0) = \lml(-y+y_0) \text{ for } y \in \realR
\end{align}
is a necessary condition for information losslessness. 

\bigskip

We
next show that $y_0$ has to equal zero to satisfy~\eqref{eq:ml_lex}
for $\beta = 1$.

From~\eqref{eq:Lvalue_maxlog} it follows that $\sign(\lml(y)) =
(-1)^{\bar{p}_1}$ for $y< a_1$  and $\sign(\lml(y)) =
(-1)^{\bar{p}_M}$ for $y > a_M$. Letting $y = a_1 - \gamma - y_0$ for
some $\gamma > 0$ in~\eqref{eq:ml_lex} gives
\begin{equation}\label{eq:lml_symmetr}
	\lml(2y_0 -a_1 +\gamma) = \lml(a_1 -\gamma).
\end{equation}
Since $a_M = -a_1$, it is easy to see that $2y_0 -a_1 +\gamma>a_M$, hence $p_1 = p_M$. Without loss of generality, $p_1$ is assumed to be zero. Let $a_k = \argmin{a \in \setS_{1}}{(2y_0 - a_1 + \gamma - a)^2}$ and $a_l = \argmin{a \in \setS_{1}}{(a_1 -\gamma - a)^2}$, i.e., the points $a_k$ and $a_l$ are the outermost constellation points labeled with one from the right and from the left, respectively. Using the introduced notation, \eqref{eq:lml_symmetr} can be rewritten as
\begin{multline}
(2y_0 -a_1  + \gamma - a_M)^2 - (2y_0 -a_1 + \gamma  - a_k)^2 \\= (a_1
- \gamma -a_1)^2 - (a_1 - \gamma -a_l)^2\EqCo
\end{multline}
which should hold for any $\gamma>0$. Using $a_M = -a_1$, $y_0$ can be expressed as
\begin{equation}\label{eq:y_0}
	y_0 = \frac{(a_k +a_l)(2a_1 +a_k - a_l - 2\gamma)}{4(a_1 + a_k)}.
\end{equation}
By definition, $y_0$ is a constant and is independent of $\gamma$,
which implies that $a_k = -a_l$ in~\eqref{eq:y_0} and $y_0 = 0$. This
in turn means that for any two neighboring zero-crossings $y_q$ and
$y_r$, the midpoint $(y_q + y_r)/2$ has to be zero, i.e., only two
zero-crossings are allowed for the pattern to be invertible.
Therefore, an invertible pattern has to be of the form $\boldsymbol{p}
= [\mathbf{0}_{n}, \mathbf{1}_{M/2}, \mathbf{0}_{M/2-n}]$ for $n =
1,\dots, M/2-1$. The only pattern of such a form that satisfies $a_k =
-a_l$ is $\boldsymbol{p}_{\mathrm{II}} = [\mathbf{0}_{M/4},
\mathbf{1}_{M/2}, \mathbf{0}_{M/4}]$, which completes the proof.
\end{proof}

\fi

\section{Proof of Theorem \ref{theorem:harmonized_gmi}}
\label{Appendix.theor:harmonized_gmi}

\newcommand{\boldb}{\boldsymbol{b}}
\newcommand{\boldB}{\boldsymbol{B}}
\newcommand{\boldy}{\boldsymbol{y}}
\newcommand{\boldY}{\boldsymbol{Y}}
\newcommand{\M}{\mathsf{M}}

The proof follows the steps of the achievable rate analysis for
multi-level coding presented in~\cite[Ch.~3]{Fabregas08_Book}. In
order to make the proof consistent with \cite[Ch.~3]{Fabregas08_Book},
we consider a decoder that operates in the probability
domain\footnote{Strictly speaking, it is the probability domain only
for exact L-values.} according to 
\begin{equation}
	\label{eq:bicm_decoder_equivalent}
\hat{\m}(\boldy) =  \argmax{\m \in \left[|\setC|\right]} q(\bm{b}(\m),
\bm{y}) \EqCo
\end{equation}
where $\m$ denotes a message, $\boldb(\m) \in \setC$ the codeword
corresponding to message $\m$, and 
\begin{equation}
	q(\bm{b}, \bm{y}) = \prod_{i=1}^{N} \prod_{j=1}^{m} q_j(b_{i,j}, y_i)\EqCo 
\end{equation}
with $q_j(b_{i,j}, y_i) = \exp(b_{i,j} l_{j}(y_i))$. Observe
that the decoder in \eqref{eq:bicm_decoder_equivalent} is equivalent
to the decoder in \eqref{eq:bicm_decoder}. 

We consider an ensemble of length-$mN$ codes $\setC$ obtained as the
Cartesian product of $m$ binary codes of length $N$ according to
$\setC = \setC_1 \times \dots \times \setC_m$.
The codewords in each code $\mathcal{C}_j$ are assumed to be composed
of \gls{iud}~bits. The codewords of the code are equiprobable, the rate
of the bit-level codes is given by $R_j =
\log_{2}(|\setC_j|)/N$ and the overall rate is $R =
\sum_{j}R_j$. In the following, we only consider the case of two bit
positions, i.e., $m=2$. The generalization to a larger number of bit
positions is straightforward. 

We let $\boldb_1(\m_1)$ and $\boldb_2(\m_2)$ denote the codewords in
$\setC_1$ and $\setC_2$ corresponding to individual messages $\m_1$ and $\m_2$,
respectively. When averaging over codebooks, the codewords become
random vectors $\boldB_1(\m_1)$ and $\boldB_2(\m_2)$, where
\begin{align}
	\Pr{\boldB_2(\m_2) = \boldb} = \Pr{\boldB_1(\m_1) = \boldb} =  \prod_{i=1}^N p_B(b_i)
\end{align}
for all $\m_1 \in [|\setC_1|]$ and $\m_2 \in
[|\setC_2|]$.\footnote{Even though in this paper the bits are
\gls{iud}, i.e., $p_B(b_i) = 1/2$, we keep the notation general to be consistent with~\cite[Ch.~3]{Fabregas08_Book}.} Due to the
fact that the code is constructed as a product code, for any given
codes $\setC_1$ and $\setC_2$ and all $\m_2 \in [|\setC_2|]$
\begin{align}
	\hat{\m}_1(\boldy) &= \argmax{\m_1 \in \left[|\setC_1|\right]}
	q([\bm{b}_1(\m_1),  \boldb_2(\m_2)], \bm{y}) \\
	&=  \argmax{\m_1 \in
	\left[|\setC_1|\right]} q_1(\bm{b}_1(m_1), \bm{y})\label{eq:decoder_decompose} \EqCo
\end{align}
and analogously for the other bit-level code. 

Since the codewords are equiprobable, the probability of error
averaged over the ensemble of randomly generated codes is given by
\begin{equation}\label{eq:gmi_pe}
	\pebar = \frac{1}{|\setC_1||\setC_2|} \sum_{\m_1 \in
	\left[|\setC_1|\right]} \sum_{\m_2 \in \left[|\setC_1|\right]}
	\pebar(\m_1, \m_2) \EqCo
\end{equation}
where $\pebar(\m_1, \m_2)$ denotes the ensemble-averaged error
probability conditional on messages $\m_1$ and $\m_2$ being
transmitted. However, due to the random code construction, the
probability of error is independent of the particular transmitted
messages and hence $\pebar = \pebar(\m_1, \m_2)$ for any given $\m_1$
and $\m_2$. 

The probability $\pebar(\m_1, \m_2)$ can be calculated as
\begin{equation}\label{eq:funcf}
\pebar(\m_1, \m_2) = \Expect{f(\boldB_1(\m_1), \boldB_2(\m_2), \boldY)} \EqCo
\end{equation}
where $f(\cdot)$ is defined in~\eqref{eq:big_phat_eq}.  For a given
observation $\boldy$, the function $f(\cdot)$
in~\eqref{eq:big_phat_eq} can be upperbounded
by~\eqref{eq:err_terms}, where~\eqref{eq:meh0} follows from the union
bound and~\eqref{eq:err_terms} follows
from~\eqref{eq:decoder_decompose}.
\begin{figure*}[!t]
\normalsize
\setcounter{equation}{71}
\begin{align}\label{eq:big_phat_eq}
&f(\boldb, \boldb', \boldy)= \Pr{\hat{\m}_1(\boldy) \neq \m_1 \cup
\hat{\m}_2(\boldy) \neq \m_2 \,|\,
	\boldB_1(\m_1) = \boldb, \boldB_2(\m_2) = \boldb', \boldY = \boldy} \\
	&\leq
	\Pr{\hat{\m}_1(\boldy) \neq \m_1  \,|\,\boldB_1(\m_1) = \boldb, \boldB_2(\m_2) = \boldb', \boldY = \boldy   }+
	\Pr{\hat{\m}_2(\boldy) \neq \m_2  \,|\,\boldB_1(\m_1) = \boldb, \boldB_2(\m_2) = \boldb', \boldY = \boldy   }\label{eq:meh0}\\
	&=
	\Pr{\hat{\m}_1(\boldy) \neq \m_1  \,|\,\boldB_1(\m_1) = \boldb, \boldY = \boldy}
	+
	\Pr{\hat{\m}_2(\boldy) \neq \m_2  \,|\,\boldB_2(\m_2) = \boldb', \boldY = \boldy}\label{eq:err_terms}.
\end{align}
\setcounter{equation}{74}
\hrulefill
\vspace*{4pt}
\end{figure*}
Using again the union bound and following the standard steps of Gallager's error analysis~\cite{Gallager1968}, the first probability
in~\eqref{eq:err_terms} can be upperbounded as
\begin{align}
	&\Pr{\hat{\m}_1(\boldy) \neq \m_1  |\boldB_1(\m_1) = \boldb, \boldY = \boldy}\nonumber\\
	& = \Pr{\bigcup_{\m' \neq \m_1} \left\{ \hat{\m}_1(\boldy) = \m'  |\boldB_1(\m_1) = \boldb, \boldY = \boldy \right\}}\nonumber\\
	& \le \left(\sum_{\m' \neq \m_1} \Pr{\hat{\m}_1(\boldy) = \m'
	|\boldB_1(\m_1) = \boldb, \boldY = \boldy}\right)^{\gamma} \label{eq:meh1}\\
	& = (|\setC_1| - 1)^{\gamma} \left(\Pr{\hat{\m}_1(\boldy) = \m'
	|\boldB_1(\m_1) = \boldb, \boldY = \boldy}\right)^{\gamma} \label{eq:meh2}
\end{align}
for all $0\le\gamma\le 1$ and any $\m' \neq \m_1$. We further upperbound the probability in~\eqref{eq:meh2} as
\begin{multline}
	\Pr{\hat{\m}_1(\boldy) = \m'  |\boldB_1(\m_1) = \boldb, \boldY = \boldy} \\
	\stackrel{}{=} \hspace{-0.8cm}\sum_{\substack{\boldb' \in \mathcal{B}^{N},\\q_1(\boldb', \boldy))\ge q_1(\boldb, \boldy))}} \hspace{-0.8cm} p_{\boldB}(\boldb')
	\le \sum_{\boldb' \in\mathcal{B}^{N}}
	p_{\boldB}(\boldb')\frac{q_1(\boldb',
	\boldy)^{s_1}}{q_1(\boldb,\boldy)^{s_1}} \EqCo
\end{multline}
where the inequality holds for any $s_1 \geq 0$ since $q_1(\boldb', \boldy))\ge
q_1(\boldb, \boldy))$ and the sum over all
codewords gives an upper bound. Substituting the obtained result
into~\eqref{eq:meh2}, we have
\begin{align} 
	&\Pr{\hat{\m}_1(\boldy) \neq \m_1  |\boldB_1(\m_1) = \boldb, \boldY = \boldy}\nonumber\\
	& \le (|\setC_1| - 1)^{\gamma} \left(\sum_{\boldb'
	\in\mathcal{B}^{N}} p_{\boldB}(\boldb')\frac{q_1(\boldb',
	\boldy)^{s_1}}{q_1(\boldb,\boldy)^{s_1}}\right)^{\gamma}.\label{eq:pr_upper_bound}
\end{align}

Averaging over all possible codewords $\boldB_1(\m_1)$ and the observations $\boldY$ gives
\begin{align} 
	&\Expect{\Pr{\hat{\m}_1(\boldy) \neq \m_1  |\boldB_1(\m_1) = \boldb, \boldY = \boldy}}\nonumber\\
	& \le \Expect{(|\setC_1| - 1)^{\gamma} \left(\sum_{\boldb'
	\in\mathcal{B}^{N}} p_{\boldB}(\boldb')\frac{q_1(\boldb',
	\boldY)^{s_1}}{q_1(\boldB,\boldY)^{s_1}}\right)^{\gamma}}.\label{eq:meh3}\\
	& = (|\setC_1| - 1)^{\gamma} \left(\Expect{\left(\sum_{b'
	\in\mathcal{B}} p_{B}(b')\frac{q_1(b',
	Y)^{s_1}}{q_1(B_1,Y)^{s_1}}\right)^{\gamma}}\right)^{N} \EqCo\label{eq:meh4}
\end{align}
where to go from~\eqref{eq:meh3} to~\eqref{eq:meh4} we used the fact
that the channel is memoryless. 

Applying similar steps to the second probability
in~\eqref{eq:err_terms} and combining~\eqref{eq:funcf},
\eqref{eq:err_terms}, \eqref{eq:meh4}, and the definition of rate for
the constituent codes, we can upperbound the probability of error as 
\begin{equation}\label{eq:pe_error}
\pebar \leq 2^{-N(E_1(\gamma, s_1) - \gamma R_1)} + 2^{-N(E_2(\gamma,
s_2) - \gamma R_2)} \EqCo
\end{equation}
where
\begin{align}
	E_j(\gamma, s_j) = -\log \Expect{\left( \sum_{b' \in \mathcal{B}} p_{B}(b')
	\frac{q_j(b', Y)^{s_j} }{ q_j(B_j, Y)^{s_j} }
	\right)^\gamma }. 
\end{align}
Observe that \eqref{eq:pe_error} holds for any choice of $0 \leq
\gamma
\leq 1$, $s_1 \geq 0$, and $s_2 \geq 0$. Furthermore, the probability of
error vanishes if $E_j(\gamma, s_j) > \gamma R_j$ for given $s_j \geq 0 $ for
$j = 1,2$. In particular, if all rates satisfy
\begin{align}\label{eq:meh5}
	R_j &< {\sup_{s_j \geq 0}} \lim_{\gamma \rightarrow
	0}{}{\frac{E_j(\gamma, s_j)}{\gamma}}\\
	&= \GMI_{L_j}\EqCo
\end{align}
we have $\pebar \to 0$ as $N \to \infty$. Evaluating the right-hand
side of~\eqref{eq:meh5} by the means of derivative gives the bit-level
\gls{GMI} in~\cite[eq.~(17)]{Nguyen11}. Hence, $\pebar \to 0$ as $N \to \infty$ if $R = \sum_j R_j < \sum_j \GMI_{L_j}$, which concludes the proof.

\section{Convexity of Mutual Information}
\label{Appendix.theor:mutual_information_convexity}
\newcommand{\setG}{\mathcal{G}}

For a fixed input distribution $f_X(x)$, the mutual information
$I(X;Y)$ is a convex function in the channel law $f_{Y|X}(y|x)$
\cite[Th.~2.7.4]{Cover2006second}. The following theorem
particularizes this result to the case of binary-input,
continuous-output channels and provides necessary and sufficient
conditions for equality. This theorem is used in the proof of
Corollaries~\ref{corol:symmetr} and~\ref{corol:mix}.

\begin{theorem}
	\label{theorem:mutual_information_convexity}
	Let $B$ be a binary RV and let $S$ be a discrete RV independent of
	$B$ taking values in $[m]$, where $m \in \mathbb{N}$. Furthermore,
	let $\{W_j(\cdot|\cdot)\}$, for $j \in [m]$, be a collection of
	binary-input, continuous-output channels with domains $\mathcal{G}_j
	= \{l \in \mathbb{R} : W_j(l|1) \neq 0 \text{ and } W_j(l|0) \neq 0
	\}$. Given $B = b$ and $S=s$, let $L$ be a continuous RV distributed
	according to $W_s(\cdot|b)$. For $j \in [m]$, define the functions
	$g_j$ as
	\begin{align}
		\label{eq:convexity_correction}
		g_j(l) = \log \left( \frac{W_j(l|1)}{W_j(l|0)} \right).
	\end{align}
	Then, given a fixed distribution on $B$ and
	$S$, we have $I(B;L) \leq I(B;L|S)$ with equality if and only if
	\begin{align}
		g_j(l) = g_{j'}(l), \qquad \text{almost everywhere on $\mathcal{G}_j
		\cap \mathcal{G}_{j'}$}
	\end{align}
	for all $j,j' \in [m]$. 
\end{theorem}

\begin{proof}

We have
\begin{align}
	\label{eq:logsum_proof_conditional_mutinf}
	I(B; L|S) = \sum_{j=1}^{m} f_S(j) I(B;L | S = j) \EqCo
\end{align}
where
\begin{align}
	&I(B;L | S = j) = \\
	& \sum_{b} f_B(b) \int_{\setG_j} W_j(l|b) \log \left(
	\frac{ W_j(l|b) }{ \sum_{b'} f_B(b') W_j(l|b')}
	\right) \, \mathrm{d} l \EqCo \label{eq:logsum_proof_mutinf}
\end{align}
Inserting \eqref{eq:logsum_proof_mutinf} into
\eqref{eq:logsum_proof_conditional_mutinf} and swapping summation and
integration, we obtain
\begin{align}
	&I(B; L|S) = \\
	&\sum_{b} f_B(b) \sum_{j=1}^m f_S(j) \int_{\setG_j} W_j(l|b) \log \left(
	\frac{ W_j(l|b) }{ \sum_{b'} f_B(b') W_j(l|b')}
	\right) \, \mathrm{d} l  \\
	&\geq \sum_b f_B(b) \int_{\bigcup_{j}\setG_j}  W(l|b) \log \left(
	\frac{ W(l|b)}{\sum_{b'} f_B(b') W(l|b')}
	\right)\, \mathrm{d} l \label{eq:logsum_proof_inequality} \\
	&= I(B; L)
\end{align}
where \eqref{eq:logsum_proof_inequality} follows from the log-sum
inequality \cite[Th.~2.7.1]{Cover2006second} and we defined
\begin{align}
	W(l|b) = \sum_{j=1}^m f_S(j) W_j(l|b). 
\end{align}
Moreover, we have equality in \eqref{eq:logsum_proof_inequality}, if and only if
for $b \in \mathbb{B}$
\begin{align}
	\frac{W_j(l|b)}{f_B(0) W_j(l|0) + f_B(1) W_j(l|1)} 
\end{align}
is independent of $j \in \{k: l \in \setG_k\}$  for almost all $l \in \bigcup_{j} \setG_j$. This condition is equivalent to the condition that $g_j(l)$ for all $j \in [m]$ are equal whenever $g_j(l)$ are defined.
\end{proof}

\balance

\end{document}